\documentclass[11pt]{article}
\usepackage[margin=0.8in]{geometry}
\usepackage[T1]{fontenc}
\usepackage[utf8]{inputenc}
\usepackage{lmodern}
\usepackage{microtype}
\usepackage{amsmath,amssymb,amsthm}
\usepackage{booktabs}
\usepackage{multirow}
\usepackage{graphicx}
\usepackage{xcolor}
\usepackage{tikz}
\usetikzlibrary{arrows.meta,positioning}
\usepackage[numbers,sort&compress]{natbib}
\usepackage[hidelinks]{hyperref}

\newtheorem{theorem}{Theorem}
\newtheorem{corollary}{Corollary}

\newcommand{\ak}[1]{{\color{black}{#1}}}

\newcommand{\rev}[1]{{\color{black}{#1}}}

\newcommand{\fs}{{\mathfrak s}}

\newcommand{\E}{\mathbb E}
\newcommand{\Rm}{\mathbb R}

\title{ML-based approach to classification and generation of structured light propagation in turbulent media}

\author{
Aokun Wang$^{1}$, Anjali Nair$^{2,*}$, Zhongjian Wang$^{1}$, Guillaume Bal$^{2,3}$\\[0.75em]
\normalsize $^{1}$Division of Mathematical Sciences, Nanyang Technological University, 21 Nanyang Link, 637371 Singapore\\
\normalsize $^{2}$Committee on Computational and Applied Mathematics, University of Chicago, Chicago, IL 60637 USA\\
\normalsize $^{3}$Departments of Statistics and Mathematics, University of Chicago, Chicago, IL 60637 USA\\
\normalsize $^{*}$\texttt{anjalinair@uchicago.edu}
}

\date{}

\begin{document}

\maketitle

\begin{abstract}
\rev{We study the classification task of structured-light beams after propagation through a random turbulent medium. 
The received speckle patterns are generated by numerical simulation of a stochastic paraxial propagation model, and the classification task is formulated over a finite alphabet of
$15$ OAM source classes. We benchmark intensity and autocorrelation inputs using SimpleCNN and ResNet-18 as classifiers. 
We also quantify the effect of training-set size and receiver-window misalignment. 
Since additional propagated samples may be costly to obtain, we develop a
class-conditioned diffusion model for generative augmentation of turbulence-degraded intensity images. 
The main contribution is a spectrum-aware diffusion objective: a pixel-domain loss combined with a Fourier-domain Bregman regularizer designed to preserve
high-frequency speckle statistics. 
We prove that this hybrid objective is consistent with the posterior-mean regression target of the diffusion model and show that generated samples substantially improve low-data classification.}

\end{abstract}

\section{Introduction} \rev{Structured light beams carrying Orbital Angular Momentum (OAM)~\cite{Willner:15, gbur2016singular, shen2019optical, borcea2020multimode, forbes2021structured} provide a rich family of spatial optical fields for free-space communication and sensing. A basic receiver-side task in this setting is to identify the transmitted structured-light symbol after propagation through a random medium. This classification problem is challenging because small-scale fluctuations in the refractive index generate complex speckle patterns that cause phase distortions, intensity scrambling, and mixing between spatial features~\cite{cox2020structured}.}

In this work, the turbulence-degraded intensity patterns are generated from numerical simulations of a paraxial It\^o--Schr\"odinger propagation model~\cite{andrews2001laser,bal2025splittingalgorithmsparaxialitoschrodinger}.
Concretely, the complex field envelope evolves under a stochastic Schr\"odinger equation with random refractive-index potential. We record the resulting intensity at the receiver plane. This forward model provides a controlled way to vary turbulence strength and to generate labeled samples associated with different OAM-mode classes.

Given each received intensity image for different realizations of the random fluctuations, we train a classifier to infer the transmitted mode class.
Deep learning methods, in particular Convolutional Neural Networks (CNNs), have become effective tools for this task~\cite{Krenn_2016,Doster:17, Avramov-Zamurovic:23, jarrett2024neural, zhang2025structured}, at least for limited turbulence effects.

Theoretical results on the structure of speckle \cite{bal2024complex, bal2025long, garnier2016fourth}, show that while scintillation (speckle strength) may be large, speckle dots may still be efficiently denoised as their spatial scale is typically small compared to the OAM-mode structure. 

The main objective of the paper is to test the classification of OAM modes under practical \rev{turbulence} intensity levels. 
We first generate labeled speckle patterns from a stochastic Schr\"odinger numerical model.
For classification, we consider two architectures: \rev{SimpleCNN} and ResNet-18. \rev{The former is less expensive to train but provides sub-optimal classification. The latter offers a good trade-off between cost and efficiency.}  \rev{Since additional optical measurements or high-resolution propagation simulations may be costly, we also study generative augmentation as a way to supplement small training sets.}
\rev{For this purpose,} we develop a physics-aware generative augmentation framework based on denoising diffusion probabilistic models (DDPMs)~\cite{ho2020denoisingdiffusionprobabilisticmodels}.
\rev{Here, physics-awareness means that the learning inputs are receiver-plane optical observables generated by a specified stochastic propagation model, and that the generative loss uses Fourier-domain speckle statistics rather than only pixel-wise image error.}

\rev{Unlike natural images, turbulent optical intensities exhibit characteristic speckle-induced high-frequency statistics~\cite{goodman2007speckle, andrews2001laser}.} We therefore combine a standard pixel-wise diffusion objective with an additional spectral consistency term that enforces the power spectral density of generated samples, and align the generative preprocessing with the downstream CNN input pipeline.

The rest of the paper is organized as follows.
Section~\ref{sec:simulation} introduces the stochastic beam-propagation model, the numerical solver, and \rev{the structured-light dataset generated for the classification experiments; the incident source alphabet follows the OAM superposition design of \cite{Avramov-Zamurovic:23}}.
Section~\ref{sec:classification} defines the classification task, the preprocessing and evaluation protocol, and the controlled comparisons for classifier design and robustness.
Section~\ref{sec:generative} presents the conditional diffusion model, the hybrid spatial-spectral training objective with theoretic proof, and the augmentation results for classification.
Section~\ref{sec:conclusion} summarizes the main finding and conclusion of the work. Additional detail on the numerical simulations and examples of applications are presented in the Supplementary Material.

\section{Beam propagation in random media}\label{sec:simulation}
This section presents the \rev{mathematical model and numerical protocol underlying the structured-light dataset used in this paper}. \rev{The incident source alphabet is taken from the OAM superposition design of \cite{Avramov-Zamurovic:23}. The propagated intensity samples used for training, validation, and testing are then generated in this work by applying the propagation model, random-medium parameters, and discretization described below.} 
We use a paraxial model for the complex field envelope given by \cite{andrews2001laser}:
\begin{equation}\label{eq:isequation}
    \big(2i k \partial_z + \Delta_x + k^2 \nu (z,x) \big) u=0,\qquad u(z=0,x)=u_0(x),
\end{equation}
where $u(z,x)$ denotes the complex beam envelope, $z$ is the longitudinal propagation coordinate, and $x\in\mathbb{R}^2$ is the transverse coordinate. The parameter $k$ denotes the carrier wavenumber, $\Delta_x$ is the transverse Laplacian, and $\nu(z,x)$ denotes the random potential. The potential is modeled by a mean-zero stationary random process with correlation length $l$ and standard deviation parameter $\sigma$.  

The initial condition $u_0(x)$ is modeled as a superposition of Laguerre-Gaussian (LG) modes, which form an orthogonal basis carrying OAM~\cite{Avramov-Zamurovic:23}.
In cylindrical coordinates \rev{$(\rho,\varphi)$}, the field distribution of a single LG mode at the source plane ($z=0$) is expressed as:
\begin{equation}
    U_{p,l}(\rho, \rev{\varphi}, 0) = \sqrt{\frac{2p!}{\pi (p+|l|)!}} \frac{1}{w_0} \left(\frac{\rho \sqrt{2}}{w_0}\right)^{|l|} L_p^{|l|}\left(\frac{2\rho^2}{w_0^2}\right) e^{-\frac{\rho^2}{w_0^2}} e^{i l \rev{\varphi}},
    \label{eq:lg_beam}
\end{equation}
where $w_0$ is the beam waist radius, $L_p^{(|l|)}$ is the generalized Laguerre polynomial with radial index $p$ and topological charge $l$.
\rev{Following \cite{Avramov-Zamurovic:23},} we construct the source alphabet by superposing four distinct basis modes: $(p,l)\in\{(0,1),(1,4),(0,-6),(1,8)\}$ (\rev{Fig.}~\ref{fig:dataset_overview}). After the superposition, the beam is normalized to unit power.
\rev{As noted in~\cite{jarrett2024neural}, performance of classifiers vary with choices of such basis modes. The present study is restricted to the finite $15$-class alphabet described above; it does not attempt to resolve arbitrary denser choices of $l$, additional radial indices $p$, or all positive and negative topological charges, which are left for future investigation.}
\paragraph{Numerical propagation via the split-step Fourier method.}
In problems involving high frequency wave propagation over long distances, the paraxial equation~\eqref{eq:isequation} is statistically well approximated by the It\^o Schr\"{o}dinger model, where the random potential is replaced by white noise along $z$~\cite{fannjiang2004scalinglimitsbeamwave, garnier2009coupled, bal2025longMMS}. We briefly summarize this scaling regime here, with further details expanded in the Supplementary Materials. \rev{Before stating the rescaled stochastic equation, we introduce the nondimensional parameters used in it.} Let $\theta\ll 1$ and $k_0=O(1)$ be non-dimensional parameters such that $k=\frac{k_0}{\theta\ell}$. For instance, when \rev{$k=10^7\,\mathrm{m}^{-1}$} and \rev{$\ell=2\times 10^{-3}\,\mathrm{m}$}, this gives $\theta=5k_0\times 10^{-5}$. Similarly for $\sigma_0=O(1)$, we assume a sufficiently weak turbulence regime of \rev{$\sigma=\sigma_0\theta^{3/2}$}. This roughly corresponds to fluctuations in the refractive index of order $10^{-7}-10^{-6}$. Anticipating such a separation of scales, we rescale the computational domain as $z\to \frac{\ell z}{\theta}, x\to \ell x$. The choice of parameters mentioned above translates to distances in units of \rev{$80\,\mathrm{m}$} along $z$ and \rev{$2\,\mathrm{mm}$} along $x$. Asymptotically, this gives the It\^o-Schr\"{o}dinger equation in rescaled coordinates as
\begin{equation*}
    \mathrm{d}u=\frac{i}{2k_0}\Delta_xu\mathrm{d}z-\frac{k_0^2\sigma_0^2}{8}R_0(0)u\mathrm{d}z+i\frac{k_0\sigma_0}{2}u\mathrm{d}B_0\,.
\end{equation*}
\rev{This can now be solved efficiently using the split-step Fourier method (SSFM)~\cite{martin1988intensity, spivack1989split, Schmidt2010}; see also \cite{bal2025splittingalgorithmsparaxialitoschrodinger} for a theoretical validation.}

We discretize the transverse variable $x\in\mathbb{R}^2$ on a uniform Cartesian grid of size $N\times N$ (with $N=2048$ in our simulations) and grid spacing $\Delta x>0$; equivalently, the computational domain has side length $L:=N\Delta x$ (with $L=64$ in our simulations).
As an operator splitting method, SSFM approximates the propagation over a small step $\Delta z$ by composing a `refraction' step (multiplication by a random phase screen) with a `diffraction' step governed by the Laplacian. The diffraction propagator is the exponential operator $e^{\frac{i\Delta z}{2k}\Delta}$, \rev{applied} by a Fourier spectral method (FFT) on the discrete grid:
\begin{equation}\label{eq:diffraction_operator_isml}
	    \exp\!\Big(\frac{i\Delta z}{2k_0}\Delta\Big)v
	    =\mathcal{F}^{-1}\!\left(e^{-i\frac{|{\bf k}|^2}{2k}\Delta z}\,\mathcal{F}\{v\}\right),
\end{equation}
where ${\bf k}=(k_1,k_2)$ ranges over the discrete Fourier grid
\begin{equation}\label{eq:fourier_grid_isml}
	    k_\ell = \frac{2\pi}{L}\,n_\ell=\frac{2\pi}{N\Delta x}\,n_\ell,
        \quad -\frac N2\leq n_\ell \leq \frac N2-1,\ 
        \ell\in\{1,2\}.
\end{equation}
We write the (first-order) split update explicitly on the discrete longitudinal grid $z_j:=j\Delta z$. Let $u^j(x)\approx u(z_j,x)$ and let $\nu^j(x)$ denote the random potential used on the step from $z_j$ to $z_{j+1}$.
The SSFM update is the composition of:
\begin{equation}\label{eq:ssfm_refraction_isml}
    v^j(x)=\exp\!\Big( \frac i2 k\nu^j(x)\Big)\,u^j(x)\qquad\text{(refraction step)},
\end{equation}
followed by
\begin{equation}\label{eq:ssfm_diffraction_isml}
    u^{j+1}(x)=\mathcal{F}^{-1}\!\left(e^{-i\frac{|k|^2}{2k_0}\Delta z}\,\mathcal{F}\{v^j\}\right)\qquad\text{(diffraction step)}.
\end{equation}
Equivalently, by substitution,
\begin{equation}\label{eq:ssfm_scheme_isml}
    u^{j+1}(x)=\mathcal{F}^{-1}\!\left(e^{-i\frac{|{\bf k}|^2}{2k}\Delta z}\,\mathcal{F}\left\{\exp\!\big(\frac{i}{2}k\nu^j(x)\big)\,u^j(x)\right\}\right).
\end{equation}
Here $\mathcal{F}$ and $\mathcal{F}^{-1}$ denote the $2$D discrete Fourier transform and its inverse on the $N\times N$ grid.

The random medium $\nu$ in ~\eqref{eq:isequation} is prescribed with Power Spectral Density (PSD)~\cite{Schmidt2010}. In the It\^o-Schr\"{o}dinger regime, the potential is effectively white noise along $z$. For the lateral direction $x$, we use a modified isotropic spectrum, scaled by the variance prefactor $\sigma_0^2$:
\begin{equation}\label{eqn:PSD}
    \Phi({\bf k}) = \frac{\sigma_0^2 l_0^2}{4\pi^2} \exp\!\left(-\frac{l_0^2\vert {\bf k}\vert^2}{4\pi^2}\right),
\end{equation}
where ${\bf k}$ represents the coordinate vector in the $2$D frequency domain, while $(\sigma_0^2,l_0)$ are the PSD prefactor and correlation-length parameter used in \rev{our} simulations.
During the implementation of ~\eqref{eq:ssfm_scheme_isml}, we generate the discrete random potential $\nu^j(x):=\nu(z_j,x)$ at each propagation step by a Fourier synthesis method consistent with the target spectrum $\Phi$.
Specifically, the Fourier coefficients $
\hat{\nu}^j({\bf k})$ are sampled as $$
\hat{\nu}^j({\bf k}) = \sqrt{\Phi({\bf k})\,\Delta z}\,\xi^j({\bf k}),
$$ 
where $\{\xi^j({\bf k})\}_k$ is a complex Gaussian white noise field with $\xi^j({\bf k})\sim\mathcal{CN}(0,1)$ (independent for each $({\bf k},j)$). 
The spatial realization is then obtained by the inverse FFT, where $\text{Re}\{\cdot\}$ extracts the real component:
\begin{equation}
\nu^j(x) = \text{Re} \left\{ \mathcal{F}^{-1} [ \hat{\nu}^j(k) ] \right\}.
\end{equation}
\rev{In all propagation simulations used to generate the learning dataset in this paper, we fix $(z,\sigma_0,l_0,k_0)=(5,1.1,1.5,0.5)$ with discretization $\Delta z=\frac{1}{32}$.} We also set the beam waist to $w_0=4$ in the same rescaled transverse units as $x$.

\paragraph{Statistics of developed speckle.}
\rev{Fig.~\ref{fig:dataset_overview} visualizes the dataset construction. The left block shows the $15$ incident source patterns obtained from the OAM superposition alphabet of \cite{Avramov-Zamurovic:23}. The right block shows representative propagated intensity patterns generated in this work from those incident beams under the random-medium setting specified above.}
This comparison makes explicit the visual transition from structured source-scale variation to turbulence-induced speckle corruption.

\begin{figure*}[t]
    \centering
    \includegraphics[width=\textwidth]{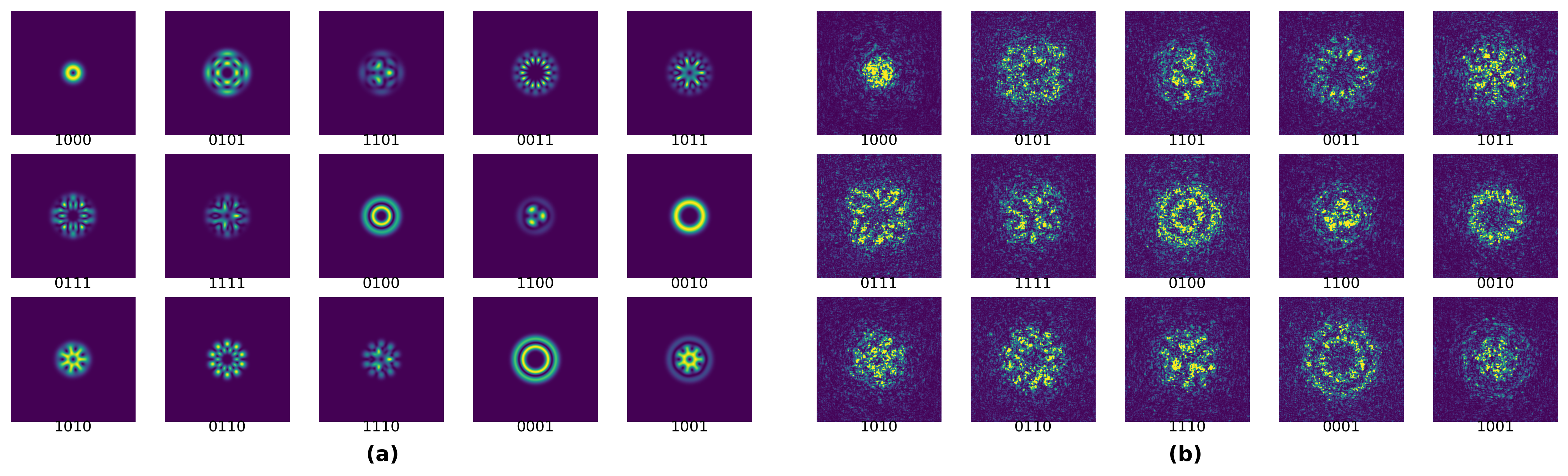}
    \caption{\rev{{Overview of the structured-light dataset generated in this work.} (a) The $15$ incident OAM superpositions, following the source-beam design of \cite{Avramov-Zamurovic:23}, used as class labels. (b) Representative propagated intensity patterns  after transmission through the random medium. Panel (b) is generated under the current default setting $(z,\sigma_0,l_0)=(5,1.1,1.5)$ using a split-step Fourier method.}     }
    \label{fig:dataset_overview}
\end{figure*}

To gain intuition on the structure of the wavebeams, we recall results on the long distance speckle patterns. In that regime, wavefields are governed by mean-zero complex Gaussian distributions, which implies exponential statistics for the intensity~\cite{garnier2016fourth, bal2024complex}. Such field distributions are uniquely determined by their correlation function obeying a diffusion equation. Specifically, the statistical average of the intensity, $\mathbb{E}I(x)$ is given by $G(z^3,\cdot)\ast I_0$, where $I_0$ is the initial intensity and $G(t,x)$ is the fundamental solution to the diffusion equation 
$$\partial_tG+\frac{1}{24}\nabla_x\cdot\Gamma\nabla_xG=0$$ 
with $G(0,x)=\delta_0(x)$.
Here, the (negative definite) diffusion tensor $\Gamma=\nabla^2R(0)$ is determined by the lateral covariance $R$ of the random medium. The scintillation index in such a setting is equal to unity.
As a natural quantifier of the strength of intensity fluctuations induced by the random medium, the scintillation index function is defined as
\begin{equation*}
    S(x):=\frac{\mathbb{E}[I^2(x)]-(\mathbb{E}[I(x)])^2}{(\mathbb{E}[I(x)])^2},
\end{equation*}
where $\mathbb{E}$ denotes expectation over turbulence realizations and $I(x)=|u(z,x)|^2$ is the received intensity.

For a plane-wave propagation ($u_0=1$), the spatially averaged empirical scintillation is found numerically to be $\overline{S} \approx 0.84$. \rev{This average is computed on the ROI $\Omega_{0.1}:=\{x:\mathbb{E}[I(x)]\ge 0.1\max_{x'}\mathbb{E}[I(x')]\}$; for the plane-wave diagnostic in the present setting this ROI occupies the full computational field. The scintillation statistic is used only to characterize the turbulence strength of the simulated dataset and is not an input to the classifier.} This corresponds to a large contribution of random speckle in the intensity maps.

\section{Classification by neural networks}\label{sec:classification}
 
We now describe the classifier input, the associated pre-processing and representation maps, and the controlled variables and their corresponding results in the classification experiments.

\subsection{Classifier with One-hot Encoding}\label{ssec:classifier}
Let $y\in\{1,\dots,15\}$ denote the class label and $e_y\in\{0,1\}^{15}$ its one-hot encoding.
Given an input image $x_{\mathrm{in}}\in\mathbb{R}^{H\times W}$, the classifier $h_\theta$ outputs a probability vector
\begin{equation}
    h_\theta(x_{\mathrm{in}})=(p_1,\dots,p_{15})\in[0,1]^{15},\qquad \sum_{i=1}^{15} p_i=1,
\end{equation}
obtained by applying a softmax map to the network logits.
We train $h_\theta$ by minimizing the cross-entropy loss on the training split and report the top-$1$ accuracy on an independent test split.
Concretely, given training samples $\{(x_{\mathrm{in}}^{(n)},y^{(n)})\}_{n=1}^N$ with class labels $y^{(n)}\in\{1,\dots,15\}$, the empirical cross-entropy objective is
\begin{equation}\label{eq:cross_entropy_isml}
    \begin{aligned}
        \mathcal{L}_{\mathrm{CE}}(\theta)
    &= -\frac{1}{N}\sum_{n=1}^N \log\Big(h_\theta(x_{\mathrm{in}}^{(n)})_{y^{(n)}}\Big)
    \\&
    = -\frac{1}{N}\sum_{n=1}^N \sum_{i=1}^{15} (e_{y^{(n)}})_i \log\big(p_i^{(n)}\big),
    \end{aligned}
\end{equation}
where $p^{(n)}=h_\theta(x_{\mathrm{in}}^{(n)})$ and $h_\theta(x_{\mathrm{in}}^{(n)})_{y^{(n)}}$ denotes the $y^{(n)}$-th component of the predicted probability vector, namely the probability assigned to the true class of the $n$-th sample.
\subsection{Neural Network Design}\label{ssec:nn_design}
This subsection specifies the two classifier architectures used in this paper; see \rev{Fig.}~\ref{fig:classifier_architecture_isml}.
We employ two convolutional architectures for comparative analysis: a lightweight SimpleCNN~\cite{lecun1998gradient}
and a standard ResNet-18~\cite{he2015deepresiduallearningimage}.
The former serves as a low-capacity baseline, while the latter represents a deeper residual network with substantially higher expressive power.
\rev{The architecture-level parameters were fixed a priori rather than selected by a hyperparameter search: SimpleCNN uses standard local convolutional blocks, and ResNet-18 follows the standard residual-stage specification with only the input and output layers adapted to one-channel images and $15$ classes.}

\begin{figure*}[t]
    \centering
    \small
    \begin{minipage}[t]{0.38\textwidth}
        \centering
        \begin{tikzpicture}[
            node distance=3.2mm,
            box/.style={draw, rounded corners, align=center, minimum height=6.5mm, minimum width=30mm},
            arr/.style={-{Latex[length=2mm]}, thick},
        ]
            \node[box] (in) {Input\\$1\times128\times128$};
            \node[box, below=of in] (b1) {Conv block 1\\$3{\times}3$, $32$ channel\\BN + ReLU + MaxPool};
            \node[box, below=of b1] (b2) {Conv block 2\\$3{\times}3$, $64$ channel\\BN + ReLU + MaxPool};
            \node[box, below=of b2] (b3) {Conv block 3\\$3{\times}3$, $128$ channel\\BN + ReLU + MaxPool};
            \node[box, below=of b3] (gap) {Global Avg.\ Pool\\$\to \mathbb{R}^{128}$};
            \node[box, below=of gap] (fc) {Linear\\$128\to 15$ logits};
            \draw[arr] (in) -- (b1);
            \draw[arr] (b1) -- (b2);
            \draw[arr] (b2) -- (b3);
            \draw[arr] (b3) -- (gap);
            \draw[arr] (gap) -- (fc);
        \end{tikzpicture}
        \vspace{1mm}\par
        \textbf{(a) SimpleCNN.}
    \end{minipage}
    \hfill
    \begin{minipage}[t]{0.58\textwidth}
        \centering
        \begin{tikzpicture}[
            node distance=3.2mm,
            box/.style={draw, rounded corners, align=center, minimum height=6.5mm, minimum width=30mm},
            arr/.style={-{Latex[length=2mm]}, thick},
        ]
            \node[box] (in) {Input\\$1\times128\times128$};
            \node[box, below=of in] (stem) {Stem\\$7{\times}7$ conv, stride $2$\\BN + ReLU + MaxPool};
            \node[box, below=of stem] (s1) {Stage 1\\$64$ channel\\BasicBlock$\times 2$};
            \node[box, below=of s1] (s2) {Stage 2\\$128$ channel\\BasicBlock$\times 2$ (downsample)};
            \node[box, below=of s2] (s3) {Stage 3\\$256$ channel\\BasicBlock$\times 2$ (downsample)};
            \node[box, below=of s3] (s4) {Stage 4\\$512$ channel\\BasicBlock$\times 2$ (downsample)};
            \node[box, below=of s4] (gap) {Global Avg.\ Pool};
            \node[box, below=of gap] (fc) {Linear\\$512\to 15$ logits};
            \draw[arr] (in) -- (stem);
            \draw[arr] (stem) -- (s1);
            \draw[arr] (s1) -- (s2);
            \draw[arr] (s2) -- (s3);
            \draw[arr] (s3) -- (s4);
            \draw[arr] (s4) -- (gap);
            \draw[arr] (gap) -- (fc);

            \begin{scope}[xshift=41mm, yshift=-15mm]
                \tikzset{rbox/.style={draw, rounded corners, align=center, minimum height=5.5mm, minimum width=22mm}}
                \node[rbox] (rin) {Input};
                \node[rbox, below=of rin] (r1) {$3{\times}3$ conv, stride $s$\\BN + ReLU};
                \node[rbox, below=of r1] (r2) {$3{\times}3$ conv, stride $1$\\BN};
                \node[rbox, below=of r2] (radd) {Add};
                \node[rbox, below=of radd] (rout) {ReLU\\Output};
                \draw[arr] (rin) -- (r1);
                \draw[arr] (r1) -- (r2);
                \draw[arr] (r2) -- (radd);
                \draw[arr] (radd) -- (rout);
                \draw[arr] (rin.east) to[out=0,in=0,looseness=1.2] node[right] {\scriptsize skip} (radd.east);
                \node[align=left, anchor=west] at ([yshift=10mm]rin.north  west) {{BasicBlock:}};
                \node[align=left, anchor=west] at ([xshift=5mm]rin.north east) {\scriptsize stride $s\in\{1,2\}$\\ \scriptsize identity skip if $s=1$;\\ \scriptsize $1{\times}1$ conv, stride $2$ if $s=2$};
            \end{scope}
        \end{tikzpicture}
        \vspace{1mm}\par
        \textbf{(b) ResNet-18.}
    \end{minipage}
    \caption{\rev{{Classifier architecture diagrams.}} Each block indicates the main computational stages from input to logits. In panel (b), the inset expands a residual \emph{BasicBlock}; here $s$ denotes the convolution stride.}
    \label{fig:classifier_architecture_isml}
\end{figure*}

\paragraph{SimpleCNN.}
As shown in \rev{Fig.}~\ref{fig:classifier_architecture_isml}(a),
it consists of three $3\times 3$ convolutional blocks with channel widths $32$, $64$, and $128$, each followed by batch normalization, ReLU, and $2\times2$ max pooling.
The final feature map is reduced by global average pooling and mapped to $15$ logits by a linear layer.
This architecture has approximately $9.5\times 10^4$ trainable parameters.

\paragraph{ResNet-18.}
As shown in \rev{Fig.}~\ref{fig:classifier_architecture_isml}(b), it involves an initial $7\times7$ convolution and max-pooling stem, followed by four residual stages with channel widths $64$, $128$, $256$, and $512$ and block counts $(2,2,2,2)$.
Each basic block contains two $3\times3$ convolutions together with a skip connection; when spatial downsampling is required, the skip branch uses a $1\times1$ projection so that the tensor shapes agree.
After the final stage, global average pooling and a linear layer produce the $15$ logits.
This architecture has approximately $1.12\times 10^7$ trainable parameters.

\rev{We also benchmarked the computational cost of the two classifiers on an NVIDIA GeForce RTX 4090 GPU, using centered $64\times64$ intensity inputs, $75$ real training samples per class, batch size $32$, and the fixed $500$-epoch training budget.
For SimpleCNN, the training time was $187.7\,\mathrm{s}$, the evaluation time on the full held-out test split of $675$ samples was $0.110\,\mathrm{s}$, and the forward-pass time for $1000$ already preprocessed inputs was $0.0067\,\mathrm{s}$.
For ResNet-18, the corresponding times were $243.0\,\mathrm{s}$, $0.156\,\mathrm{s}$, and $0.0142\,\mathrm{s}$.
Thus ResNet-18 is substantially larger and slower than SimpleCNN, although its batched forward-pass cost remains small in the tested GPU setting.}

\paragraph{Training setup.}
\rev{Unless stated otherwise, the classification experiments use the following default setting: the propagated dataset generated in this work uses $z=5$, $\sigma_0=1.1$, and $l_0=1.5$;} the classifier input is the cropped intensity on a $64\times64$ window; convolution uses zero padding; training and validation use the centered crop (equivalently $S=0$); 
testing uses the central-window protocol; and the training set contains $N_{\mathrm{train}}=50$ samples per class.
For each experiment, we train from scratch by minimizing cross-entropy on a stratified split of the \rev{propagated} dataset.
The overall proportions of samples are $50\%$ for training, $20\%$ for validation, and $30\%$ for testing.
From the training subset we subsample to obtain exactly $N_{\mathrm{train}}$ labeled samples per class.
\rev{Optimization uses Adam with weight decay $10^{-5}$, batch size $32$, and a fixed budget of $500$ epochs; the learning rate is the only training parameter selected by an automated procedure, and is chosen per architecture by a standard learning-rate range test \cite{smith2017cyclicallearningratestraining}. Unless stated otherwise, all classifier accuracy tables report top-1 classification accuracies, in percent,
evaluated on the held-out test split.}
\rev{For these classifier tables,} we select the checkpoint with the highest validation accuracy for test-time evaluation and report means and standard deviations over the three random seeds (42, 100, 2023)
used in the controlled comparisons.
\subsection{Preprocessing, Cropping, and Controlled Shifts}\label{ssec:preprocessing}
The simulated intensity fields are produced on a $2048\times 2048$ grid (Section~\ref{sec:simulation}), which are too large for convolutional processing in real applications.
For $x_{\mathrm{raw}}\in\mathbb{R}^{2048\times 2048}$, we first apply a deterministic downsampling operator
\begin{equation}\label{eq:downsample}
    D:\mathbb{R}^{2048\times 2048}\to\mathbb{R}^{256\times 256},\qquad x:=D(x_{\mathrm{raw}}),
\end{equation}
thereby defining a $256\times 256$ `canvas' image $x$.
\rev{These preprocessing dimensions are not claimed to be globally optimal: the $2048\times2048$ grid is used to resolve the propagated speckle field, the $256\times256$ canvas gives a tractable learning input after deterministic average pooling, and the later $64\times64$ crop defines a fixed finite receiver window for controlled comparisons.}
\rev{Sensitivity to crop misalignment is quantified below, and additional comparisons of fixed/random shifts and convolutional padding are reported in the Supplementary Material.}

\paragraph{Cropping Operator and Normalization.}
From the normalized canvas $\hat{x}\in\mathbb{R}^{256\times 256}$, we extract a square window of side length $W\in\mathbb{N}$.
Let $\delta=(\delta_x,\delta_y)\in\mathbb{Z}^2$ be a pixel offset relative to the centered crop.
With $(u,v)\in\{0,\dots,W-1\}^2$ and centered top-left corner
\[
\Big(\frac{256-W}{2},\,\frac{256-W}{2}\Big),
\]
we define $C_\delta:\mathbb{R}^{256\times 256}\to\mathbb{R}^{W\times W}$ by
\begin{equation*}
    \big(C_\delta(\hat{x})\big)(u,v)=\hat{x}\Big(u+\frac{256-W}{2}+\delta_x,\;v+\frac{256-W}{2}+\delta_y\Big).
\end{equation*}
In all experiments, $\delta$ is restricted so that the crop lies inside the $256\times 256$ canvas.
This cropping step has three roles: (i)
reduce computational cost; (ii) 
model partial observation: in practical sensing pipelines one may only access a windowed portion of the received intensity due to finite aperture, sensor placement, or downstream region-of-interest selection; (iii) 
make the classification task more challenging by removing potentially informative regions, thereby enabling a controlled study of robustness to information loss and spatial misalignment. 

\paragraph{Dataset-level Normalization.}
We normalize $x$ using dataset-level constants $(\mu,\sigma)$ computed once from the training split and held fixed thereafter:
$\hat{x}=\frac{x-\mu}{\sigma}$,
where $\mu\in\mathbb{R}$ and $\sigma\in\mathbb{R}_{>0}$ are scalars since the input is single-channel.
This removes an arbitrary scale from the input and stabilizes optimization across different turbulence strengths and realizations.
We compute $\mu$ and $\sigma$ by aggregating all pixels over all training images after the deterministic pre-processing stage consisting of down-sampling \eqref{eq:downsample} and the crop defined above.

\paragraph{Shifting.}
To study robustness to spatial misalignment, we vary the crop offset $\delta$ according to either a fixed shift or a random shift, with shift magnitude $S\in\mathbb{N}$ measured in pixels.
In fixed-shift mode, training uses the deterministic offset $\delta=(S,S)$.
In random-shift mode, training samples use $\delta_x,\delta_y$ sampled independently and uniformly from $\{-S,\dots,S\}$.
Unless stated otherwise, we keep the training and validation protocols identical.

\subsection{Input Field}\label{ssec:representation}
We consider two input fields for training purposes: the cropped intensity and the autocorrelation function (ACF).

\paragraph{Cropped Intensity.}
For the intensity field, the classifier input is the cropped window $x_{\mathrm{in}} = C_\delta(\hat{x})\in\mathbb{R}^{W\times W}$.

\paragraph{Autocorrelation Function (ACF).}
For the ACF field, the classifier input is the spatial autocorrelation of the cropped window.
Given $w\in\mathbb{R}^{W\times W}$, we define a mean-subtracted autocorrelation. Let $\overline{w}$ be the empirical spatial mean of $w$ given by $\overline{w}=\frac{1}{W^2}\sum_{u=0}^{W-1}\sum_{v=0}^{W-1} w(u,v)$
and let $\mathcal{F}$ denote the $2$D discrete Fourier transform.
We form the unnormalized correlation map
\begin{equation}\label{eq:acf_raw_isml}
    \Gamma_w:= \mathcal{F}^{-1}\big(|\mathcal{F}(w-\overline{w})|^2\big)\in\mathbb{R}^{W\times W},
\end{equation}
and define the ACF as:
\begin{equation}\label{eq:acf_def_isml}
    \mathrm{ACF}(w)=\frac{\Gamma_w}{\Gamma_w(0,0)}.
\end{equation}
By construction, $\mathrm{ACF}(w)(0,0)=1$. ACF maps are significantly more stable statistically than intensity maps but tend to smooth out valuable wavebeam features~\cite{katz2014non}. Table~\ref{tab:shift_study_isml} compares the two input representations. For both architectures, the cropped intensity input is consistently more accurate than the ACF input. This indicates that the classifiers can suppress part of the speckle noise while still exploiting wavebeam features that are weakened in the ACF maps. Therefore, we present intensity-based classifications for the remainder of the paper.

\begin{table}[t]
\centering
\small
\setlength{\tabcolsep}{4pt}
\caption{Input-representation comparison under the default setting. }
\label{tab:shift_study_isml}
\begin{tabular}{lcc}
\toprule
\textbf{Input} & \textbf{SimpleCNN} & \textbf{\rev{ResNet-18}} \\
\midrule
intensity & \rev{$92.59\pm1.27$} & \rev{$93.48\pm1.18$} \\
ACF       & \rev{$85.48\pm1.41$} & \rev{$86.62\pm1.74$} \\
\bottomrule
\end{tabular}
\end{table}

\subsection{Shift Study Under Controlled Misalignment}\label{ssec:shiftstudy}

A complementary comparison between deterministic fixed-shift training and random-shift training at matched $S$ is deferred to Supplementary Material Section D, Table 4.
\rev{We focus on random shifts in the main text, since they model non-systematic receiver misalignment and perform more robust than fixed shifts in the controlled comparison.}
In Table~\ref{tab:shift_magnitude_random_isml}, we report how the shift magnitude $S$ affects accuracy when the training and validation crops are sampled randomly within the range determined by $S$.
For SimpleCNN, the accuracy is nearly unchanged from $S=0$ to $S=16$, but then decreases sharply as the shift range increases further, which indicates that larger crop perturbations remove class-discriminative local structure for this architecture. 
\rev{For ResNet-18,} a small random-shift range slightly improves the mean accuracy at $S=16$, whereas larger shift magnitudes degrade performance.
\rev{Thus, in the present regime the ResNet-18, as a deeper network, remains more robust than SimpleCNN across all tested shift magnitudes, although its accuracy also decreases for large crop perturbations.}

\begin{table}[t]
\centering
\small
\setlength{\tabcolsep}{4pt}
\caption{Sensitivity to the shift magnitude $S$ under random-shift training and centered testing, with all other parameters at their default values.}
\label{tab:shift_magnitude_random_isml}
\begin{tabular}{lcc}
\toprule
\textbf{Shift magnitude $S$} & \textbf{SimpleCNN} & \textbf{\rev{ResNet-18}} \\
\midrule
0  & \rev{$92.59\pm1.27$} & \rev{$93.48\pm1.18$} \\
16 & \rev{$79.51\pm3.77$} & \rev{$93.73\pm1.13$} \\
32 & \rev{$64.05\pm5.69$} & \rev{$91.41\pm2.57$} \\
48 & \rev{$53.98\pm4.46$} & \rev{$88.15\pm3.60$} \\
\bottomrule
\end{tabular}
\end{table}

\subsection{Sensitivity to the Number of Training Sets}\label{ssec:trainsize}
We next study the dependence of classifier accuracy on the number $N_{\mathrm{train}}$ of training samples per class.
This controlled comparison quantifies sample efficiency: for each seed, we subsample the training set to $N_{\mathrm{train}}$ examples per class without replacement, when available, while keeping the validation and test sets unchanged.
Table~\ref{tab:trainsize_intensity_isml} reports the resulting accuracy for the default intensity baseline under the common evaluation protocol.
\rev{The accuracy increases substantially with $N_{\mathrm{train}}$.}
\rev{At $N_{\mathrm{train}}=25$, SimpleCNN slightly outperforms ResNet-18, suggesting that the larger residual architecture does not benefit from its higher capacity in this data-limited regime yet.}
\rev{This is also consistent with the possibility that the more complex network may overfit in the low-data regime.}
\rev{As $N_{\mathrm{train}}$ increases, ResNet-18 improves more rapidly and outperforms SimpleCNN at $N_{\mathrm{train}}=50$ and $75$.}

\begin{table}[t]
\centering
\small
\setlength{\tabcolsep}{4pt}
\caption{Training-set-size study for the default intensity baseline, with all parameters fixed at their default values except $N_{\mathrm{train}}$.}
\label{tab:trainsize_intensity_isml}
\begin{tabular}{lcc}
\toprule
$N_{\mathrm{train}}$ (per class) & SimpleCNN & \rev{ResNet-18} \\
\midrule
25 & \rev{$82.86\pm2.36$} & \rev{$79.85\pm1.29$} \\
50 & \rev{$92.59\pm1.27$} & \rev{$93.48\pm1.18$} \\
75 & \rev{$95.75\pm0.94$} & \rev{$97.19\pm0.53$} \\
\bottomrule
\end{tabular}
\end{table}

\section{Improved Classification by Augmented data}
\label{sec:generative}
This section introduces a conditional diffusion model used to generate synthetic turbulence-degraded intensity images for data augmentation in the OAM-classification task.
We specify the preprocessing used to define the generator input domain, the diffusion formulation with its prediction parameterizations, and a spectrum-aware training objective designed to preserve high-frequency speckle statistics.
\rev{The motivation is to reduce the cost of enlarging a scarce training set: once trained on propagated samples, the conditional generator can produce additional class-conditioned speckle images for classifier training.}
\rev{The generated samples are used only as a supplement; validation and testing remain restricted to propagated samples generated by the physical model.}

\subsection{Preprocessing and Normalization for Generation}
\label{ssec:gen_preproc}
Let $x_{\mathrm{raw}}\in\mathbb{R}^{2048\times 2048}$ denote a raw simulated intensity field (Section~\ref{sec:simulation}).
Training a diffusion model directly at this resolution is computationally prohibitive, hence we first downsample to a $256\times256$ image.
Concretely, we use the same deterministic downsampling operator $D$ as in ~\eqref{eq:downsample} and write $x:=D(x_{\mathrm{raw}})$ which implemented via average pooling in our code.
To train the model on a bounded domain, we rescale intensities as:
\begin{equation}\label{eq:dm_minmax_isml}
    \tilde{x} = 2\cdot\frac{x-x_{\min}}{x_{\max}-x_{\min}}-1,
\end{equation}
where $x_{\min}\in\mathbb{R}$ and $x_{\max}\in\mathbb{R}$ ($x_{\min}<x_{\max}$) are dataset-level scalars computed once from the training set by scanning all pixels across all training images after the downsampling.
Generated samples can then be mapped back to the original intensity scale for downstream evaluation.

\subsection{Conditional diffusion model and parameterizations}\label{ssec:gen_ddpm}
We employ a Denoising Diffusion Probabilistic Model (DDPM) conditioned on the OAM class $c\in\{1,\dots,15\}$.
Let $x_0$ denote a normalized training sample.
The forward diffusion process produces noisy states $x_t$ according to a variance schedule $(\beta_t)_{t=1}^T$:
\begin{equation}\label{eq:forward_ddpm_isml}
    x_t = \sqrt{\bar{\alpha}_t}\,x_0 + \sqrt{1-\bar{\alpha}_t}\,\epsilon,\qquad \epsilon\sim\mathcal{N}(0,I).
\end{equation}
where $\bar{\alpha}_t=\prod_{s=1}^t (1-\beta_s)$.
The reverse model is parameterized by a neural network $f_\theta(x_t,t,c)$ whose output is compared to a prediction target $y_{\mathrm{target}}$.
\rev{Concretely, during training we sample a real propagated image $x_0$ with class label $c$, diffusion time $t$, and Gaussian noise $\epsilon$.}
\rev{The noised image $x_t$ is then formed by \eqref{eq:forward_ddpm_isml}, and the network is trained to predict a target determined by the chosen parameterization.}
\rev{For example, if the sampled image belongs to class $c=7$ and the training configuration is $v$-prediction, then the network input is the triple $(x_t,t,7)$ and the regression target is the velocity variable $v_t$ defined in \eqref{eq:vpred_isml}.}
We consider three standard parameterizations, corresponding to different choices of the regression target.
First, in $\epsilon$-prediction, the network predicts the forward noise: $y_{\mathrm{target}}=\epsilon.$
Second, in $x_0$-prediction, the network predicts the clean sample: $y_{\mathrm{target}}=x_0.$
Third, in $v$-prediction, the network predicts the velocity variable: $y_{\mathrm{target}}=v_t$,
where 
\begin{equation}\label{eq:vpred_isml}
    v_t := \sqrt{\bar{\alpha}_t}\,\epsilon - \sqrt{1-\bar{\alpha}_t}\,x_0.
\end{equation}
The $v$-parameterization is often numerically stable across noise levels and is well-suited to settings where high-frequency structure must be preserved during sampling \cite{song2021scorebasedgenerativemodelingstochastic}.

\paragraph{U-Net Architecture and Training Details.}
We implement $f_\theta$ as a class-conditioned U-Net with six resolution levels and two convolutional layers per resolution block.
Concretely, the channel widths follow the multiplier pattern, $(128,256,384,512,512,512)$ channels, from high to low resolution, and we include self-attention at the lowest resolutions to capture long-range dependencies in the speckle patterns.
Unless stated otherwise, we train using 
AdamW on a $70/30$ train/validation split with initial learning rate $10^{-4}$, using batch sizes $8$ for training and $30$ for validation, for a fixed budget of $200$ epochs, and we retain the checkpoint with the lowest validation loss for sampling.
We use a ReduceLROnPlateau schedule with multiplicative factor $0.5$, patience $10$ epochs and minimum learning rate $10^{-9}$.

\subsection{Loss space}
We use the mean squared error in the \rev{pixel} domain
\begin{equation}\label{eq:dm_pixel_loss_isml}
\mathcal{L}_{\mathrm{pixel}}=\mathbb{E}\Big[\big\|f_\theta(x_t,t,c)-y_{\mathrm{target}}\big\|^2\Big],
\end{equation}
where the expectation is taken over the training data, the forward noise $\epsilon$, and the time index $t$.
\rev{This is a target-regression problem conditional on $(x_t,t,c)$: for each noisy image, diffusion time, and class label, the network estimates the corresponding denoising target.}
\rev{The choice of target and the choice of loss space are related but distinct.}
\rev{The target determines what the network output represents, while the loss space determines after which linear reparameterization the prediction error is measured.}
In light of the recent preprint \cite{li2025back} and the earlier work \cite{lu2025mathematical}, the selection of the norm $\|\cdot\|$ fundamentally affects the learning despite their equivalence under invertible linear transforms \eqref{eq:vpred_isml}; similar concepts appear in the numerical analysis theory of preconditioning~\cite{xu1992iterative}. Corresponding to the three choices of target, we define three norms in \eqref{eq:dm_freq_loss_isml}, denoted as $\|\cdot\|_{\mathbf{v}}$, $\|\cdot\|_{\boldsymbol{\epsilon}}$, and $\|\cdot\|_{\mathbf{x}}$. For instance, given $y_{\mathrm{target}}=\mathbf{x}_0$, under $\|\cdot\|_{\mathbf{v}}$ as loss space, \eqref{eq:dm_pixel_loss_isml} reads
\begin{equation}
    \mathcal{L}_{\mathrm{pixel}}
    =
    \mathbb{E}\Bigg[
        \Bigg\|
        \frac{\sqrt{\bar{\alpha}_t}\,x_t-f_\theta(x_t,t,c)}{\sqrt{1-\bar{\alpha}_t}}
        -
        \frac{\sqrt{\bar{\alpha}_t}\,x_t-x_0}{\sqrt{1-\bar{\alpha}_t}}
        \Bigg\|_2^2
    \Bigg],
\end{equation}
which is the mean-squared error after decoding the network output into the $\mathbf{v}$-parameterization used in the implementation.

\subsection{Hybrid Spatial-spectral Training Objective}\label{ssec:gen_loss}
Speckle intensities exhibit prominent high-frequency statistics, while pixel-wise regression losses are known to favor low-frequency components in many image models \cite{rahaman2019spectralbiasneuralnetworks}.
To make the generator sensitive to spectral content without changing the basic diffusion training pipeline, we use a hybrid objective
\begin{equation}\label{eq:dm_total_loss_isml}
    \mathcal{L}_{\mathrm{total}}=\mathcal{L}_{\mathrm{pixel}}+\lambda\,\mathcal{L}_{\mathrm{freq}},
\end{equation}
with $\lambda\geq0$. In all the experiments presented later, we take $\lambda=1$ as default setting.

\paragraph{Frequency-domain Regularizer.}
Let $\mathcal{F}$ denote the two-dimensional discrete Fourier transform.
We set $\rho(z)=|z|$ (applied element-wise) and define the frequency-domain regularizer by
\begin{equation}\label{eq:dm_freq_loss_isml}
    \mathcal{L}_{\mathrm{freq}}
    =\mathbb{E}\!\left[
        D_\rho\!\left(\mathcal{F}(y_{\mathrm{target}}),\,\mathcal{F}\big(f_\theta(x_t,t,c)\big)\right)
    \right],
\end{equation}
where $D_\rho$ is defined in ~\eqref{eq:bregmandivergence_isml}.
\rev{Here $\rho$ is a fixed convex penalty applied to each Fourier coefficient, and $\partial\rho$ denotes its subdifferential.}
\rev{For the choice $\rho(z)=|z|$, this subdifferential is $\partial\rho(z)=\{z/|z|\}$ when $z\neq0$ and $\partial\rho(0)=\{\xi:|\xi|\leq1\}$ at the origin.}
\rev{These objects are part of the loss definition and do not depend on the training or test split.}
\rev{The sample-dependent quantities are the Fourier coefficients of the target and of the network prediction.}
\rev{With $\rho(z)=|z|$, the regularizer compares spectral amplitudes rather than directly comparing complex Fourier phases.}
\rev{This is useful for speckle images because the power distribution over spatial frequencies is physically meaningful, whereas a global phase factor in the Fourier representation is not directly relevant to the received intensity statistics.}
\rev{Thus this term penalizes spectral discrepancies while remaining agnostic to a global phase factor.}
Similar frequency-domain losses without the Bregman divergence have also been proposed in image restoration and generation contexts; see, e.g., \cite{jiang2021focalfrequencylossimage}.
In comparison, we can establish Bayes-consistency of the proposed hybrid objective \eqref{eq:dm_freq_loss_isml}.
At the population level, we show that minimizing the hybrid loss recovers the conditional expectation of the regression target given $(x_t,t,c)$.
In the diffusion setting, this corresponds to posterior-mean optimality for the chosen prediction parameterization.

We first note that the pixel-domain term already satisfies this requirement.
Under the mild regularity condition $\mathbb{E}[\Vert y_{\mathrm{target}}\mid x_t,t,c \Vert]<\infty$, the squared loss $\mathbb{E}[\Vert f_\theta(x_t,t,c)-y \Vert^2_2\mid x_t,t,c]$ is uniquely minimized by $f^*(x_t,t,c)=\mathbb{E}[y_{\mathrm{target}}\mid x_t,t,c]$.

It remains to show that adding the frequency-domain regularizer does not change this minimizer.
We formalize this using a Bregman divergence defined in a transformed space (here, the Fourier domain):
\begin{equation}\label{eq:bregmandivergence_isml}
    D_\rho(x,y) = \rho(x)-\rho(y)-\mathrm{Re}\langle\xi_y, x-y\rangle, \qquad \xi_y\in\partial\rho(y)
\end{equation}
where $\partial\rho(y)$ denotes the subdifferential of $\rho$ at $y$ (and if $\rho$ is differentiable at $y$, then $\xi_y=\nabla\rho(y)$).
\rev{Intuitively, the Bregman divergence measures the excess value of the convex penalty at the target relative to the first-order affine approximation of that penalty at the prediction.}
\rev{In the present application, this construction gives a way to compare the target and predicted Fourier-domain quantities while retaining a conditional-mean optimality property.}
The formal definition leads to the following result.

\begin{theorem}[Consistency of Bregman Divergence Minimization]\label{thm:bregman_consistency_isml}
Let $x_0 \in \mathbb{R}^d$ be the target signal and $x_t$ be the noisy observation. Let $\mu = \mathbb{E}[x_0 | x_t]$ be the posterior mean. Let $\mathcal{T}: \mathbb{R}^d \to \mathbb{C}^d$ be an invertible linear operator, and let $\rho(\cdot)$ be a \textbf{convex} function (e.g., the $L_1$ norm).
Define the objective function $J(f)$ as the expected Bregman divergence induced by $\rho$:
\begin{equation}
    J(f) = \mathbb{E}_{x_0 | x_t} \left[ \sum_{k=1}^d D_\rho \left( (\mathcal{T}{x}_0)_k, (\mathcal{T}f({x}_t))_k \right) \right],
\end{equation}
where $D_\rho(a, b)$ is defined in ~\eqref{eq:bregmandivergence_isml} (applied component-wise).
Then, the conditional expectation $f^*(x_t) = \mu$ is a \textbf{global minimizer} of the objective function $J(f)$.
\end{theorem}
\begin{proof} 
Fix the observation ${x}_t$ and let ${\mu}=\mathbb{E}[{x}_0|{x}_t]$.
Let $\mathbf{U}=\mathcal{T}{x_0}$ with conditional mean ${\nu}=\mathcal{T}{\mu}$.
The objective decomposes component-wise: $J(f) = \sum_kg_k(w_k)$, where $w_k=(\mathcal{T}f({x}_t))_k$ and 
\begin{equation*}
    g_k(w_k) = \mathbb{E}\!\left[D_{\rho}(U_k,w_k)\,\middle|\,x_t\right].
\end{equation*}
We decompose the loss at an arbitrary prediction \rev{$w_k$} against the loss at the mean $\nu_k$.
The difference is:
\begin{equation*}
     g_k(w_k) - g_k(\nu_k) = \mathbb{E}\!\left[D_{\rho}(U_k,w_k) - D_\rho(U_k,\nu_k)\,\middle|\,x_t\right].
\end{equation*}
We expand the Bregman divergence terms:
\begin{equation*}
    \begin{aligned}
        &D_\rho(U_k,w_k) - D_{\rho}(U_k,\nu_k) 
         \\=& \left(\rho(\nu_k)-\rho(w_k)-\text{Re}\langle\xi_{w_k},U_k-w_k\rangle+\text{Re}\langle\xi_{\nu_k},U_k-\nu_k\rangle\right).
    \end{aligned}
\end{equation*}
Now, we take the expectation with respect to $U_k$. Using the linearity of expectation and the fact that $\mathbb{E}[U_k]=\nu_k$:
\begin{equation*}
    \begin{aligned}
        \mathbb{E}[\text{Re}\langle\xi_{w_k},U_k-w_k\rangle] &= \text{Re}\langle\xi_{w_k},\mathbb{E}[U_k]-w_k\rangle
        = \text{Re}\langle\xi_{w_k},\nu_k-w_k\rangle
    \end{aligned}
\end{equation*}
Similarly, the term involving $\xi_{\nu_k}$ vanishes because $\mathbb{E}[U_k-\nu_k]=0$.
Substituting these expectation back, we get:
\begin{equation*}
    g_k(w_k) - g_k(\nu_k) = \rho(\nu_k)-\rho(w_k)-\text{Re}\langle\xi_{w_k},\nu_k-w_k\rangle
\end{equation*}
Recognizing the structure of the right-hand side, this is exactly the Bregman divergence $D_\rho(\nu_k,w_k)$:
\begin{equation*}
    g_k(w_k)-g_k(\nu_k) = D_{\rho}(\nu_k,w_k)
\end{equation*}
By the definition of Bregman divergence for a convex function $\rho$, we have $D_\rho(\nu_k,w_k)\geq 0 $.
Therefore, $g_k(w_k)\geq g_k(\nu_k)$ for all $w_k$.
This confirms
that $\nu_k$ minimizes the component-wise loss.
Transforming back to the spatial domain, $f^*({x}_t)=\mathcal{T}^{-1}{\nu}={\mu}$ is a global minimizer.
\end{proof}
This implies both $\mathcal{L}_{\mathrm{pixel}}$ and $\mathcal{L}_{\mathrm{freq}}$ admit the same conditional expectation as the unique global minimizer. The optimality is invariant under linear transformation of the prediction target and loss space induced by \eqref{eq:vpred_isml}, and we summarize the following corollary.
\begin{corollary}[Consistency of the Hybrid Objective]\label{cor:hybrid_consistency_isml}
Fix $(x_t,t,c)$ and let $Y$ denote the pixel-domain training target conditioned on $(x_t,t,c)$, so that $\mathcal{L}_{\mathrm{pixel}}$ corresponds to the conditional loss of $Y$.
Assume further that, for the chosen DDPM parameterization, $x_0$ is an affine function of $(x_t,Y)$ at fixed $t$ (this includes $x_0$-, $\epsilon$-, and $v$-prediction).
Then the population objective corresponding to ~\eqref{eq:dm_total_loss_isml} is uniquely minimized by the posterior mean estimator $\mathbb{E}[x_0\mid x_t,t,c]$.
\end{corollary}
\rev{Corollary \ref{cor:hybrid_consistency_isml} implies that adding an arbitrary scale of $\mathcal{L}_{\mathrm{freq}}$ to the original loss function does not alter the population posterior-mean minimizer of the target-regression problem.}
\rev{This statement should be distinguished from the finite-sample numerical effect measured below.}
\rev{The theorem shows that the Fourier-domain Bregman term is statistically compatible with the DDPM regression target.}
\rev{In finite-sample neural-network training, the additional term changes the optimization emphasis by penalizing spectral mismatch in the generated speckle images.}

\subsection{Generative Augmentation Protocol and Results}\label{ssec:gen_aug}
We now describe the controlled augmentation experiment used to assess whether generated samples improve low-data classification.
Let $N_{\mathrm{real}}$ denote the number of real training samples per class and $N_{\mathrm{syn}}$ denote the number of synthetic samples per class drawn from a trained conditional diffusion model.
In the present study we fix $N_{\mathrm{real}}=25$ and $N_{\mathrm{syn}}=50$, so that each augmented training set contains $75$ samples per class in total.
The classifier is then trained on the union of the real and synthetic samples, while validation and testing continue to use only real data.
\rev{The Real-75 baseline measures the alternative of adding the same number of additional propagated samples.}

The synthetic images are sampled from four diffusion-model configurations.
\rev{Each configuration is indexed by an ordered pair consisting of the network prediction target and the training loss space.}
\rev{For instance, $\mathbf{v}$-pred / $\mathbf{x}$-loss means that the network is parameterized to predict the velocity variable $v_t$, but the prediction and target are decoded to the clean-image variable $x_0$ before the loss is evaluated.}
\rev{After training a generator, we sample a normalized synthetic image by running the reverse diffusion chain conditioned on a prescribed class $c$.}
\rev{We then invert the min--max normalization~\eqref{eq:dm_minmax_isml} to return the sample to the physical intensity scale.}
\rev{Finally, the synthetic intensity is passed through the same classifier preprocessing pipeline as a real propagated image: downsampling/cropping, training-set normalization, and inclusion in the classifier training set with its conditioning class as label.}
This controlled protocol ensures that any change in classifier accuracy is attributable to the synthetic augmentation itself rather than to a mismatch in preprocessing conventions.

\begin{table}[t]
\centering
\caption{\rev{Classification accuracy (\%) under data scarcity and generative augmentation for both classifier architectures. `Gen.\ config' denotes the diffusion prediction target / loss space used for synthetic sampling.}}
\label{tab:dm_aug_results_isml}
\footnotesize
\setlength{\tabcolsep}{2.4pt}
\begin{tabular}{@{}llcc@{}}
\toprule
\textbf{Data setting} & \textbf{Gen.\ config} & \textbf{\rev{ResNet-18}} & \textbf{\rev{SimpleCNN}} \\
\midrule
\rev{Real-25} & -- & \rev{$79.85 \pm 1.29$} & \rev{$82.86 \pm 2.36$} \\
\rev{Real-75} & -- & \rev{$97.19 \pm 0.53$} & \rev{$95.75 \pm 0.94$} \\
\rev{Real-25 + Syn-50} & \rev{$\mathbf{x}$-pred  / $\mathbf{v}$-loss} & \rev{$89.68 \pm 0.82$} & \rev{$78.91 \pm 0.67$} \\
\rev{Real-25 + Syn-50} & \rev{$\mathbf{v}$-pred  / $\mathbf{x}$-loss} & \rev{$91.75 \pm 2.08$} & \rev{$85.78 \pm 2.10$} \\
\rev{Real-25 + Syn-50} & \rev{$\mathbf{v}$-pred  / $\mathbf{v}$-loss} & \rev{$92.54 \pm 1.07$} & \rev{$86.77 \pm 0.75$} \\
\rev{Real-25 + Syn-50} & \rev{$\mathbf{v}$-pred  / $\boldsymbol{\epsilon}$-loss} & \rev{$90.72 \pm 0.99$} & \rev{$87.46 \pm 1.07$} \\
\rev{Syn-75} & \rev{$\mathbf{x}$-pred  / $\mathbf{v}$-loss} & \rev{$71.06 \pm 2.08$} & \rev{$25.09 \pm 5.80$} \\
\rev{Syn-75} & \rev{$\mathbf{v}$-pred  / $\mathbf{x}$-loss} & \rev{$88.49 \pm 1.72$} & \rev{$40.79 \pm 5.49$} \\
\rev{Syn-75} & \rev{$\mathbf{v}$-pred  / $\mathbf{v}$-loss} & \rev{$87.95 \pm 0.31$} & \rev{$74.77 \pm 2.89$} \\
\rev{Syn-75} & \rev{$\mathbf{v}$-pred  / $\boldsymbol{\epsilon}$-loss} & \rev{$85.43 \pm 0.99$} & \rev{$68.79 \pm 3.82$} \\
\bottomrule
\end{tabular}
\end{table}

\rev{For the $\lambda$-sensitivity study in Table~\ref{tab:lambda_comparison_simplecnn_isml}, only the frequency-loss weight is varied. 
The classifier is SimpleCNN and the generator uses the $\mathbf{v}$-prediction / $\mathbf{v}$-loss configuration. 
As for the details of dataset, the training set is Real-25 + Syn-50, and validation and testing use only real propagated samples.}
\rev{For the corresponding $\lambda=1$ mixed-data setting, Figures~\ref{fig:simplecnn_lambda1_confusion_isml} and~\ref{fig:resnet18_lambda1_confusion_isml} report the class-wise pooled confusion matrices for SimpleCNN and ResNet-18, respectively.}

\begin{table}[t]
\centering
\caption{\rev{Sensitivity to the frequency-loss weight $\lambda$ for SimpleCNN under the $\mathbf{v}$-prediction / $\mathbf{v}$-loss configuration. The data setting is Real-25 + Syn-50 under the default propagation setting $(z,\sigma_0,l_0)=(5,1.1,1.5)$, with real validation and test splits.}}
\label{tab:lambda_comparison_simplecnn_isml}
\footnotesize
\begin{tabular}{lcccc}
\toprule
\textbf{$\lambda$} & $0$ & $10^{-1}$ & $1$ & $10$ \\
\midrule
\textbf{Accuracy (\%)}    & \rev{$85.98 \pm 0.82$} & \rev{$87.16 \pm 0.82$} & \rev{$86.77 \pm 0.75$} & \rev{$78.27 \pm 1.54$} \\
\bottomrule
\end{tabular}
\end{table}

\rev{Table~\ref{tab:dm_aug_results_isml} shows that, under the present setting, synthetic augmentation substantially improves on the Real-25 baseline for ResNet-18, although it remains below the Real-75 baseline.}
\rev{For ResNet-18, the best default-$\lambda$ result is obtained with '$\mathbf{v}$-pred / $\mathbf{v}$-loss', which increases the mean accuracy from $79.85\%$ to $92.54\%$; the closely related '$\mathbf{v}$-pred / $\mathbf{x}$-loss' configuration gives a comparable result.}
\rev{The same table gives the corresponding SimpleCNN comparison.}
\rev{The SimpleCNN column is consistent with the main ResNet-18 conclusion in the following sense: mixed real--synthetic training with the $\mathbf{v}$-prediction generators outperforms the Real-25 baseline, but it does not reach the Real-75 baseline.}
\rev{For SimpleCNN, the '$\mathbf{v}$-pred / $\mathbf{v}$-loss' setting increases the mean accuracy from $82.86\%$ to $86.77\%$, while the '$\mathbf{v}$-pred / $\boldsymbol{\epsilon}$-loss' setting gives a comparable result at $87.46\%$.}
\rev{The '$\mathbf{x}$-pred / $\mathbf{v}$-loss' configuration is an exception for SimpleCNN, since it falls below the Real-25 baseline, indicating that the benefit of generated samples depends on both the generator parameterization and the classifier architecture.}
\rev{When the classifier is trained only on generated samples, the best Syn-75 result reaches $88.49\%$ for ResNet-18 and $74.77\%$ for SimpleCNN, both below the corresponding mixed real--synthetic training results.}

\rev{Table~\ref{tab:lambda_comparison_simplecnn_isml} further shows that, within the SimpleCNN $\mathbf{v}$-prediction / $\mathbf{v}$-loss setting, moderate frequency-domain regularization improves the classification accuracy relative to $\lambda=0$. The choices $\lambda=0.1$ and $\lambda=1$ give comparable accuracies within the empirical variation over seeds, while $\lambda=10$ leads to a clear degradation.}
\rev{This suggests that a moderate frequency-domain penalty is not harmful in this protocol, whereas an overly large value can overweight spectral matching relative to the pixel-space diffusion objective.}
\begin{figure}[t]
    \centering
    \includegraphics[width=0.9\linewidth]{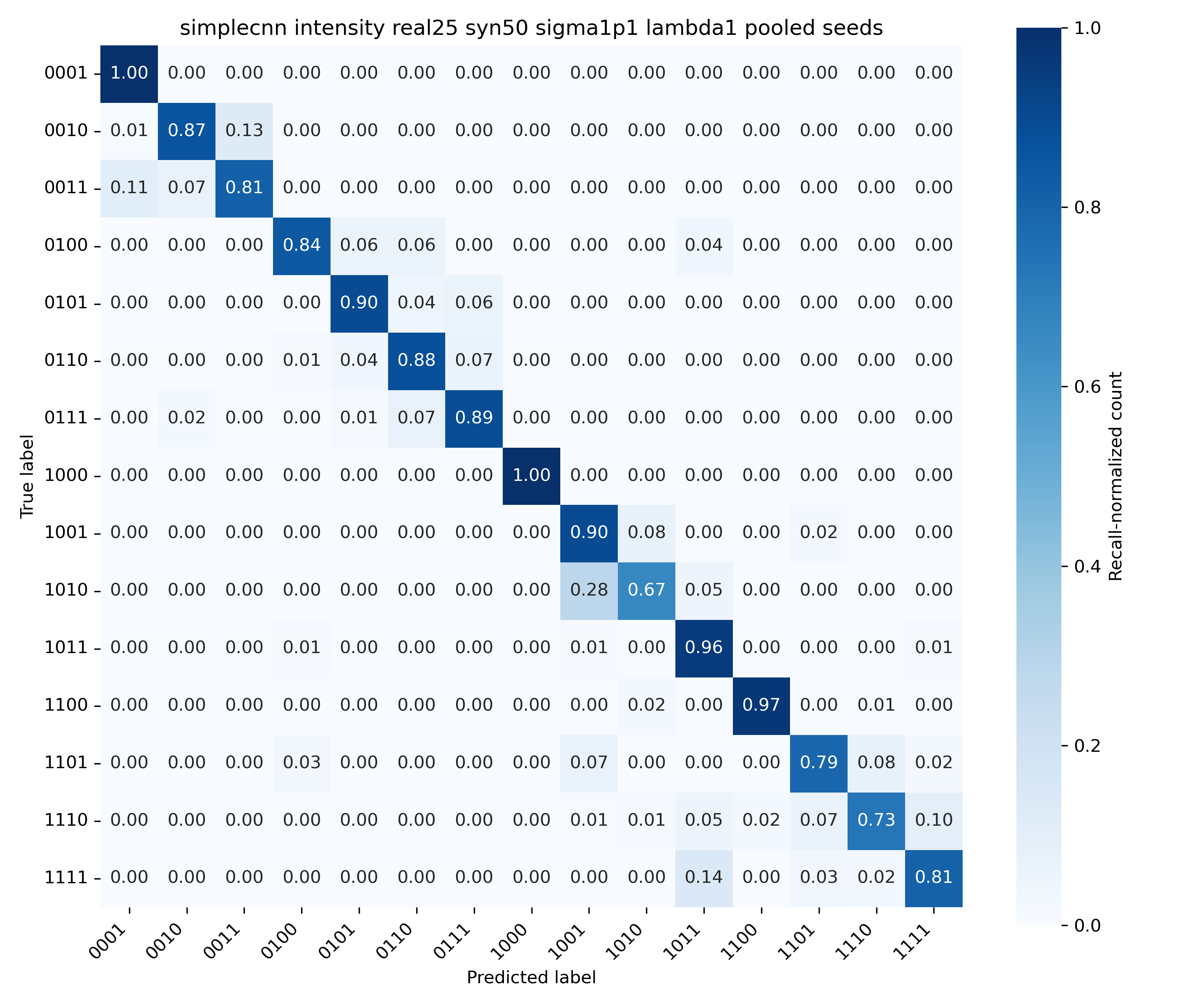}
    \caption{\rev{Pooled row-normalized confusion matrix for SimpleCNN under the Real-25 + Syn-50, $\mathbf{v}$-prediction / $\mathbf{v}$-loss, $\lambda=1$ setting. The matrix pools predictions from the three independently trained seeds used in Table~\ref{tab:lambda_comparison_simplecnn_isml}.}}
    \label{fig:simplecnn_lambda1_confusion_isml}
\end{figure}

\begin{figure}[t]
    \centering
    \includegraphics[width=0.9\linewidth]{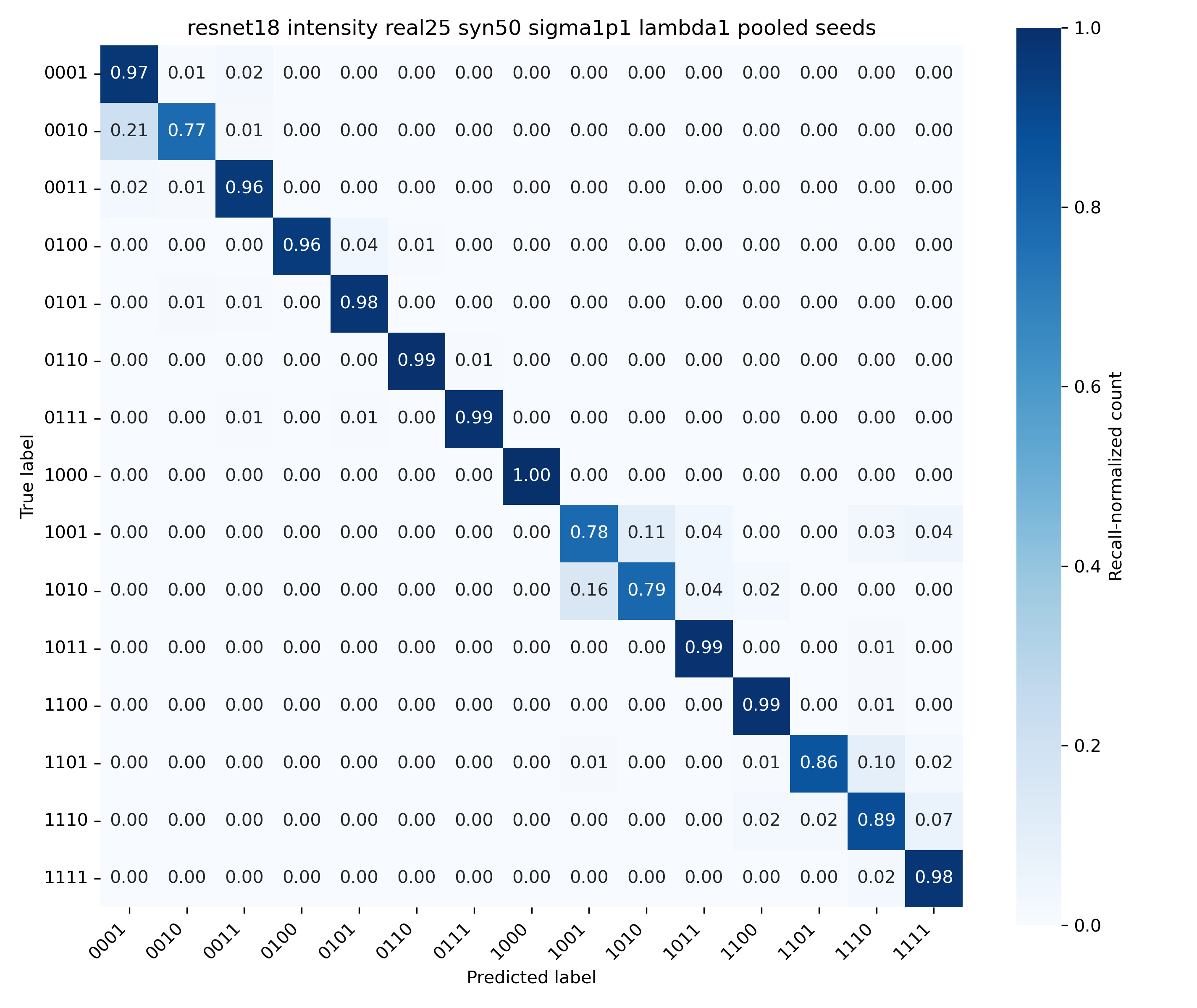}
    \caption{\rev{Pooled row-normalized confusion matrix for ResNet-18 under the Real-25 + Syn-50, $\mathbf{v}$-prediction / $\mathbf{v}$-loss, $\lambda=1$ setting. The matrix pools predictions from the three independently trained seeds used in Table~\ref{tab:dm_aug_results_isml}.}}
    \label{fig:resnet18_lambda1_confusion_isml}
\end{figure}

\rev{Within the present controlled protocol, the generated data provide a useful supplement in the low-data regime, but do not replace additional real samples.} Representative generated samples across configurations, together with a direct visual comparison between a real propagated sample and a generated sample, are deferred to \rev{Supplementary Material}.

\section{Findings and Conclusions}\label{sec:conclusion}
\rev{In this work, we benchmarked standard convolutional classifiers for structured-light classification and developed a diffusion-based augmentation method for low-data training.} \rev{We first specified the source alphabet, following \cite{Avramov-Zamurovic:23}, and generated propagated speckle patterns with our split-step Fourier propagation protocol.} We next tested neural networks for classifying $15$ beam classes from local Intensity and ACF observations. In Section~\ref{sec:classification}, we found that \emph{ResNet-18 with intensity as input} outperforms other choices. We also validated that increasing the training data size and keeping the observation window centered improved classification performance.  To target the practical limitations of available observation data, we developed in Section~\ref{sec:generative} a generative model based on the denoising diffusion probabilistic model with a new  \emph{consistent} hybrid training objective that improves generation quality in the high-frequency domain. \rev{The Bregman distance used in this objective 
is physics-motivated through the Fourier-domain structure of turbulent speckle intensities, whose high-frequency spectral content is important for generating useful training samples.} Augmented data from the generative models did improve classification rates. \rev{Finally, among possible preconditionings of the network approximation target and loss function, we found that \emph{$\mathbf{v}$-prediction and $\mathbf{v}$-loss} achieved optimal performance.}

\section*{Acknowledgments} The authors would like to thank Svetlana Avramov-Zamurovic for many enlightening discussions on the classification of structured light. This work was funded in part by ONR Grant N00014-26-1-2017 and NSF Grant DMS-2306411. ZW was partly supported by NTU SUG-023162-00001, MOE AcRF Tier 1 Grant RG17/24 and SPMS Collaborative Research Award.

\section*{Disclosures} The authors declare no conflicts of interest.

\section*{Data Availability Statement} Data underlying the results presented in this paper are available in \href{https://github.com/aokun023/IS-ML-classifier-code}{https://github.com/aokun023/IS-ML-classifier-code}.

\bibliographystyle{unsrtnat}
\bibliography{citation}

\appendix
\section{Supplementary Analyses and Numerical Details}
The remaining material collects the numerical details, validation diagnostics, and auxiliary controlled comparisons that support the main text.

\subsection{Wavebeam Propagation Model.}\label{supp:propmodel} We start from the standard paraxial approximation to the Helmholtz wave equation \cite{andrews2001laser}
\begin{equation}
  \big(2i k \partial_z + \Delta_x + k^2 \nu (z,x) \big) u=0,\qquad u(z=0)=u_0.
\end{equation}
Here $k=2\pi/\lambda=\omega/c$ is the effective wavenumber and $\nu(z,x)$ modeling the random fluctuations is a Gaussian field to simplify the analysis. More specifically, we assume that
\[
  \nu(z,x) = \sigma \nu_0(\frac z \ell, \frac x \ell)
\]
for a strength parameter $\sigma$, a correlation length $\ell$ and a fixed Gaussian field $\nu_0$. Then, for an isotropic stationary field
\[
  C(z,x) = \E \nu(z,x) \nu(0,0) = \sigma^2 C_0 ( \frac z \ell, \frac x \ell),\quad R(x) = \int_{\Rm} C(z,x) dz = \sigma^2 \ell R_0(\frac x \ell),
\]
where $C_0(z,x) = \E \nu_0(z,x) \nu_0(0,0)$ and $R_0=\int_{\Rm} C_0(z,x)dz$. 
Asymptotically for $0<\tau\ll 1$, as a central limit approximation,
\[
  \int_0^\Delta \tau^{-\frac12} \nu_0(\frac z\tau,x) dz \approx B_0(\Delta,x),\qquad \E B_0^2(\Delta,x) = R_0(x) \Delta.
\]
This is justified for $\ell\ll L$, where $L$ is the distance of propagation. Very formally, this is equivalent to
\[
 \nu_0(z,x) = \frac{d}{dz} B_0(z,x), \qquad \nu(z,x) = \sigma \ell ^{\frac12} \frac{d}{dz} B_0(z,x).
\]
In this limit, the paraxial model is well-approximated by the following Stratonovich-Schr\"odinger stochastic partial differential equation~\cite{fannjiang2004scalinglimitsbeamwave, garnier2009coupled, bal2025longMMS}:
\[
  2ik du + \Delta_x udz  + k^2 \ell ^{\frac12} \sigma dB_0 (z,\frac x\ell)\circ u=0.
\]
With the standard stochastic correction from Stratonovich to It\^o calculus, this gives in It\^o form
\[
  du = \frac{i}{2k} \Delta_x u dz - \frac{k^2 \ell \sigma^2 R_0(0)}{8} u dz + \frac {ik}2 \ell^{\frac12} \sigma u dB_0(z,\frac x\ell).
\]
Defining $dB(z,x) = \sigma \ell^{\frac12} dB_0(z,\frac x\ell)$ and $R(x) = \sigma^2 \ell R_0(\frac x\ell)$, now with $R$ and $dB$ having units of distance, we have
\begin{equation}\label{eq:ItoSch}
  du = \frac{i}{2k} \Delta_x u dz - \frac{k^2  R(0)}{8} u dz + \frac {ik}2  u dB(z,x).
\end{equation}

The statistical second moment $\mu(z,x,y) = \E u(z,x) u^*(z,y)$ is shown by It\^o stochastic calculus \cite{garnier2016fourth} to satisfy the equation:
\[
  \partial_z \mu = \frac{i}{2k} (\Delta_x-\Delta_y) \mu  + \frac14 k^2  (R(x-y)-R(0)) \mu.
\]
The second moment for incident plane waves, i.e., $u_0=1$ is thus given by the explicit formula
\[
 \mu(z,x,y) =  e^{\frac 14 k^2 z (R(x-y)-R(0))}.
\]

The speckle size may therefore be \ak{modeled} as the value $\fs$ such that, for instance for a correlation decay of $e^{-\frac12}$, we have 
\begin{equation}\label{eq:specklesize}
    \frac 14 k^2 z (R(0) - R(\fs)) = \frac12.
\end{equation}
Assuming $\fs$ is small enough that $R(0)-R(\fs) \approx \frac12|\nabla^2 R(0)| |\fs|^2$, then we find
\begin{align*}
  |\fs|^2 &= \frac{4}{|\nabla^2 R(0)| k^2 z} = \frac{4\ell^2}{\sigma^2 |\nabla^2 R_0(0)|  k^2 \ell z} = \frac{4}{|\nabla^2 R_0(0)|} \frac{Z}z \eta^2 \ell^2, \\
  \eta^2 &:=\frac{1}{\sigma^2 k^2 \ell^2} \frac{\ell}{Z},
\end{align*}
where $Z$ is a characteristic distance of interest such that $\eta$ is smaller than $1$ for the above approximation by the Hessian to be valid.

Using $R_0(x)=e^{-|x|^2}$ for instance, we get
\begin{equation}\label{eq:specklesize2}
    \fs = \eta \ell \sqrt{\frac {2Z}z}
     = \frac{\sqrt{2}\ell}{\sigma k \sqrt {\ell z}}.
\end{equation}

Consider $k=10^7 m^{-1}$, $\ell= 2\, 10^{-3}m$, $\sigma=10^{-6}$, $Z=500m$, we find
\[ \eta = \frac{1}{2\ 10^{-2}} 2\ 10^{-3}=\frac1{10}.\]
For $z=2Z=10^3m$, we find a speckle size $\fs=10^{-1} \ell =2 \, 10^{-4}m$.

\subsection{Numerical Implementation.}\label{supp:numimple}
All simulations are performed on rescaled \ak{dimensionless} spatial variables $(x,y)\in[-32,32]^2$ discretized by $2048\times 2048$ pixels. The rescaled axial variable is modeled by the interval $[0,z_0]$ with axial step size $\Delta z = 1/32$ in the default setting.

The spatial $(x,y,z)$ variables are rescaled as follows. For non-dimensional parameters $\theta\ll 1$ and $k_0=O(1)$, define $\theta=\frac{k_0}{k\ell}$. Let $\sigma_0=\frac{\sigma}{\theta^\frac32}$ and $w_0=\frac{w}{\ell}$. Here $\sigma_0$ is again a standard-deviation strength parameter, so the corresponding variance prefactor is $\sigma_0^2$, consistent with the simulation convention in the main text. We compress the computational domain as $z\to \frac{\ell z}{\theta}, x\to \ell x$.  This gives the It\^o-Schr\"{o}dinger equation in rescaled coordinates as
\begin{equation*}
    \mathrm{d}u=\frac{i}{2k_0}\Delta_xu\mathrm{d}z-\frac{k_0^2\sigma_0^2}{8}R_0(0)u\mathrm{d}z+i\frac{k_0\sigma_0}{2}u\mathrm{d}B_0,
    \end{equation*}
with $u(0,x)=u_0\big(\frac{x}{w_0}\big)$.
Let $z_0$ be the propagation distance for the above \ak{nondimensionalized} equation. When \ak{$k=10^7\,\mathrm{m}^{-1}$ and $\ell=2\times 10^{-3}\,\mathrm{m}$}, this gives $\theta=5k_0\times 10^{-5}$. For $k_0=1/2$, this gives total propagation distance as \ak{$80z_0\,\mathrm{m}$} and $\sigma=1.25 \sigma_0\times 10^{-7}$. 
\ak{In free-space communication through air,} $\sigma\approx 10^{-7}-10^{-6}$ and total distance $z$ ranges from \ak{a few hundred meters} to at most a couple of kilometers. Varying $\sigma_0$ between 1 and 4 while varying $z_0$ between 4 and 8 gives rise to this setting, in which case, $\eta$ as defined in the previous section for plane wave incident beams varies from $0.1-1$.

For beams with width $w_0$, the speckle size $\fs$ we expect to observe numerically, using the same method as above in ~\eqref{eq:specklesize}, is
\begin{equation*}
    \frac{k_0^2\sigma_0^2z_0|\fs|^2|\nabla^2R_0(0)|}{32}=\frac{1}{2},\quad \mbox{ i.e., }\quad
    |\fs|=\frac{4}{\sqrt{k_0^2\sigma_0^2z_0|\nabla^2R_0(0)|}}\,.
\end{equation*}
The pixel size $\Delta x$ should ideally be small enough \ak{to} resolve this. This is the \ak{microscale behavior}. 

There is also diffusion at the macroscopic scale, which broadens the beam \cite{bal2024complex}. If the initial condition is Gaussian for instance, $u_0(x)=e^{-\frac{|x|^2}{2w_0^2}}$, solving a diffusion model \cite{bal2024complex} to compute the mean intensity gives
\begin{align*}
    \mathbb{E}I(z_0,r)=&\Big|\mathbb{I}_d-\frac{\sigma_0^2z_0^3\nabla^2R_0(0)}{12w_0^2}\Big|^{1/2}
    \\ &\times
    \exp\Big(-\frac{1}{2w_0^2}r^\top\big[\mathbb{I}_d-\frac{\sigma_0^2z_0^3\nabla^2R_0(0)}{12w_0^2}\Big]r\Big)\,.
\end{align*}
So we may take the beam broadening rate ($\alpha=w(z_0)/w_0$) to be
\begin{equation*}
    \alpha=\sqrt{1+\frac{\sigma_0^2|\nabla^2R_0(0)|z_0^3}{12w_0^2d}}\,.
\end{equation*}
In a uniform medium, the spreading rate is proportional to $\frac{z_0}{k_0w_0^2}$, which becomes more dominant when the above is small. So, the simulation domain should be proportional to $w_0\alpha$. When the correlation function $R_0(x)=\frac{1}{4\pi}e^{-\pi^2|x|^2}$, we have $|\nabla^2R_0(0)|=\pi$ for $d=2$. This gives the crude estimate $\Delta x\ll |\fs|\approx\frac{5}{\sqrt{\sigma_0^2z_0}}$ and box length $2L$ with $L\gg \alpha w_0\approx \sqrt{w_0^2+0.1\sigma_0^2z_0^3}$\,.

For convenience, Table~\ref{tab:simulation_parameters_isml} summarizes the default numerical parameters used for the simulations reported in the main text.

\begin{table}[htbp!]
\centering
\small
\setlength{\tabcolsep}{6pt}
\caption{Simulation parameters used in the numerical experiments, unless explicitly stated otherwise.}
\label{tab:simulation_parameters_isml}
\begin{tabular}{ll}
\toprule
\textbf{Quantity} & \textbf{Value} \\
\midrule
Transverse computational domain & $[-32,32]^2$ \\
Transverse grid size & $2048\times 2048$ \\
Transverse grid spacing & $\Delta x = 64/2048 = 1/32$ \\
Axial step size & $\Delta z = 1/32$ \\
Default propagation distance & $z_0=5$ \\
Lateral correlation $R_0(x)$ & $\frac{1}{4\pi}\exp(-\pi^2|x|^2/l_0^2)$\\
Default turbulence strength & \ak{$\sigma_0 = 1.1$} \\
Default correlation-length parameter & $l_0 = 1.5$ \\
Beam waist at the source & $w_0 = 4$ \\
Source alphabet & $(p,l)\in\{(0,1),(1,4),(0,-6),(1,8)\}$ \\
Number of beam classes & $15$ \\
\bottomrule
\end{tabular}
\end{table}

\paragraph{Speckle-resolution Diagnostics via Autocorrelation.}

We ensure that individual speckle grains are reasonably resolved numerically. To quantify this, we first define
\[
I_0(x):=I(x)-\langle I\rangle,
\]
where $\langle I\rangle$ denotes the empirical average of $I$ over the discrete $N\times N$ pixel grid.
The normalized spatial autocorrelation of $I_0$ is then
\begin{equation}\label{eq:acf_spatial_isml}
    p(\delta):=\frac{\sum_x I_0(x)\,I_0(x+\delta)}{\sum_x I_0(x)^2},
\end{equation}
where $\delta\in\mathbb{Z}^2$ denotes a discrete two-dimensional pixel shift, the sum is taken over all grid points, and the shifted index $x+\delta$ is understood periodically on the discrete grid.
With this normalization, $p(0)=1$.
Numerically, ~\eqref{eq:acf_spatial_isml} is evaluated efficiently via
\begin{equation}\label{eq:acf_periodic_isml}
    p(\delta)=\frac{\mathcal{F}^{-1}\!\left(\left|\mathcal{F}(I_0)\right|^2\right)(\delta)}{\mathcal{F}^{-1}\!\left(\left|\mathcal{F}(I_0)\right|^2\right)(0)},
\end{equation}
where $\mathcal{F}$ denotes the discrete Fourier transform on the periodic grid.
To reduce this two-dimensional function to a one-dimensional profile, we define $\bar p(r)$ by averaging $p(\delta)$ over all shifts $\delta$ with the same Euclidean radius $r:=|\delta|$, where $r$ is measured in pixels.
\begin{equation}\label{eq:acf_radial_isml}
    \bar p(r):=\operatorname{avg}_{|\delta|=r} p(\delta).
\end{equation}
We then define the correlation length by
\begin{equation}\label{eq:corrlen_isml}
    l_c:=\min\{r\ge 1:\bar p(r)<e^{-1}\},
\end{equation}
with the convention $l_c=r_{\max}$ if $\bar p(r)\ge e^{-1}$ for all $r\le r_{\max}$.
For the current default setting \ak{$(z,\sigma_0,l_0)=(5,1.1,1.5)$}, computed from $n=150$ realizations of the OAM code $1101$ on the $2048\times2048$ grid, we obtain
\[
    \ak{l_c = 67.61 \pm 8.85 \quad \text{pixels}},
\]
with observed values ranging from \ak{$47$ to $87$ pixels}. The propagated speckle pattern is thus well resolved numerically.
The corresponding radial ACF profile, histogram, and summary statistics are reported in \ak{Fig.~\ref{fig:acf_validation_isml}} and Table~\ref{tab:acf_validation_isml}.

\subsection{Numerical Validation.}
\label{supp:numvalid}

To ensure that the simulated intensity patterns are numerically well-resolved, we perform a frequency-domain validity check based on the spatial sampling theorem.
Recall that $\Delta x$ denotes the transverse grid spacing of the $N\times N$ simulation grid introduced above.
The corresponding Nyquist frequency is: $f_{\text{Nyq}} = \frac{1}{2\Delta x}$, 
which defines the highest spatial frequency that can be represented without aliasing on the discrete grid.
A necessary condition for numerical fidelity is that spatial frequency content of the simulated fields decays sufficiently before reaching $f_{\text{Nyq}}$.

For the propagated intensity
$I(x)=|u(z,x)|^2$,
we define its discrete spatial power spectral density (PSD) by
\begin{equation}\label{eq:intensity_psd_isml}
    P_I(k)=\big|\mathcal{F}\{I\}(k)\big|^2.
\end{equation}
To assess numerical resolution, we examine the radially averaged profile of ~\eqref{eq:intensity_psd_isml}. For each radius $r\ge 0$, let
\begin{equation}\label{eq:radial_psd_isml}
    \overline{P}_I(r):=\frac{1}{|\mathcal K_r|}\sum_{k\in \mathcal K_r} P_I(k),
\end{equation}
where $\mathcal K_r$ denotes the set of discrete frequency samples whose Euclidean radius lies in the radial bin centered at $r$.
We then plot the normalized radial PSD
\begin{equation}\label{eq:normalized_radial_psd_isml}
    \widetilde{P}_I(r):=\frac{\overline{P}_I(r)}{\max_{r'}\overline{P}_I(r')}.
\end{equation}
This image PSD $P_I$ should be distinguished from the medium spectrum $\Phi$ given by
\begin{equation}\label{eqn:PSD}
    \Phi({\bf k}) = \sigma_0^2 l_0^2 \exp\!\left(-\frac{l_0^2\vert {\bf k}\vert^2}{4\pi^2}\right),
\end{equation}
which prescribes the random potential $\nu$.
As shown in \ak{Fig.~\ref{fig:psd_check}}, the normalized radial PSD for a representative OAM class is concentrated well below the Nyquist limit, with negligible mass near $f_{\text{Nyq}}$, which indicates that the intensity field is adequately resolved on the computational grid.
\begin{figure}[htbp!]
    \centering
    \includegraphics[width=0.6\linewidth]{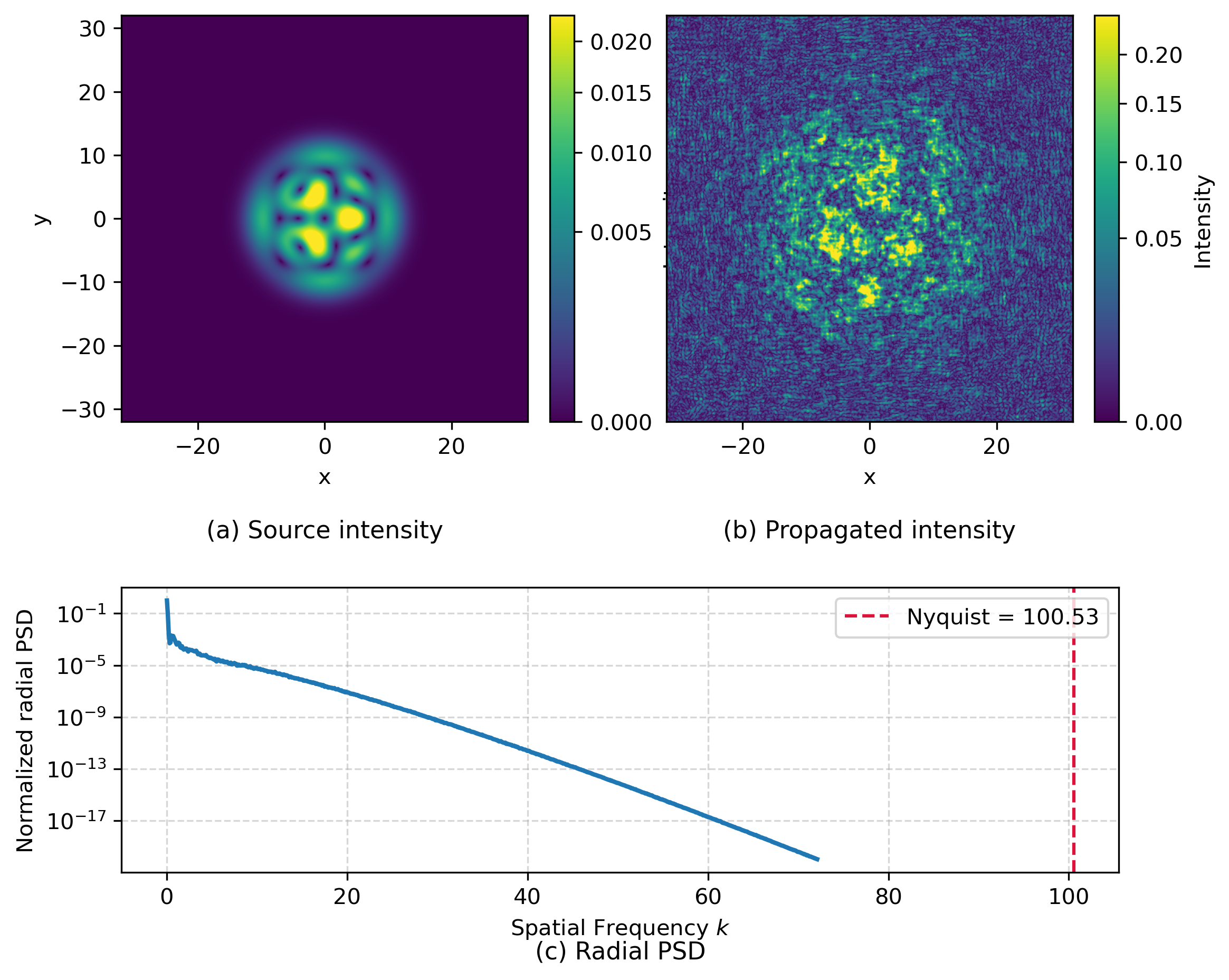}
    \caption{Frequency-domain resolution check for the OAM code $1101$ \ak{under the default setting $(z,\sigma_0,l_0)=(5,1.1,1.5)$}. \ak{Panel (a) shows the source intensity associated with the code $1101$; panel (b) shows a representative propagated intensity $|u(z,\cdot)|^2$ generated with the current medium spectrum~\eqref{eqn:PSD}; and panel (c) shows the normalized radial profile $\widetilde P_I(r)$ of the image PSD~\eqref{eq:intensity_psd_isml}, with the Nyquist frequency marked by the dashed vertical line.} The spectral mass is concentrated well below the Nyquist limit.}
    \label{fig:psd_check}
\end{figure}

We complement this frequency-domain check with an autocorrelation-based diagnostic using the correlation-length definition \eqref{eq:corrlen_isml}. For the current default setting \ak{$(z,\sigma_0,l_0)=(5,1.1,1.5)$}, we compute the radial ACF profile and the corresponding correlation length for all $n=150$ realizations of the OAM code $1101$. \ak{Fig.~\ref{fig:acf_validation_isml}} shows the empirical mean and standard deviation of the radial ACF together with the histogram of the resulting correlation lengths, and Table~\ref{tab:acf_validation_isml} reports the corresponding summary statistics. In this setting, the estimated correlation length is \ak{$l_c=67.61\pm 8.85$ pixels}, which confirms that the propagated intensity remains well resolved on the $2048\times 2048$ grid.
\begin{figure}[htbp!]
    \centering
    \begin{tabular}{cc}
    \includegraphics[width=0.35\linewidth]{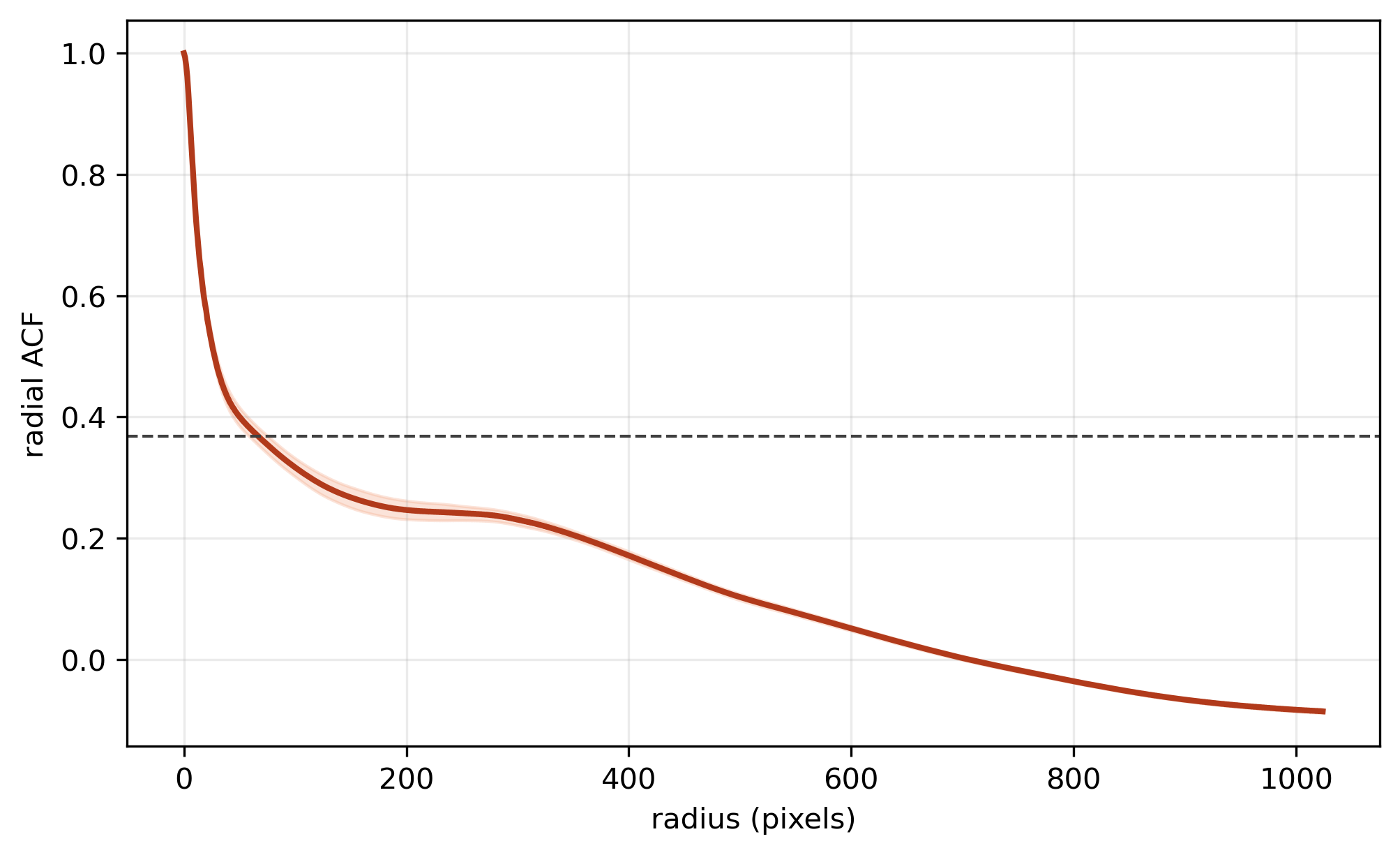}
    &
    \includegraphics[width=0.35\linewidth]{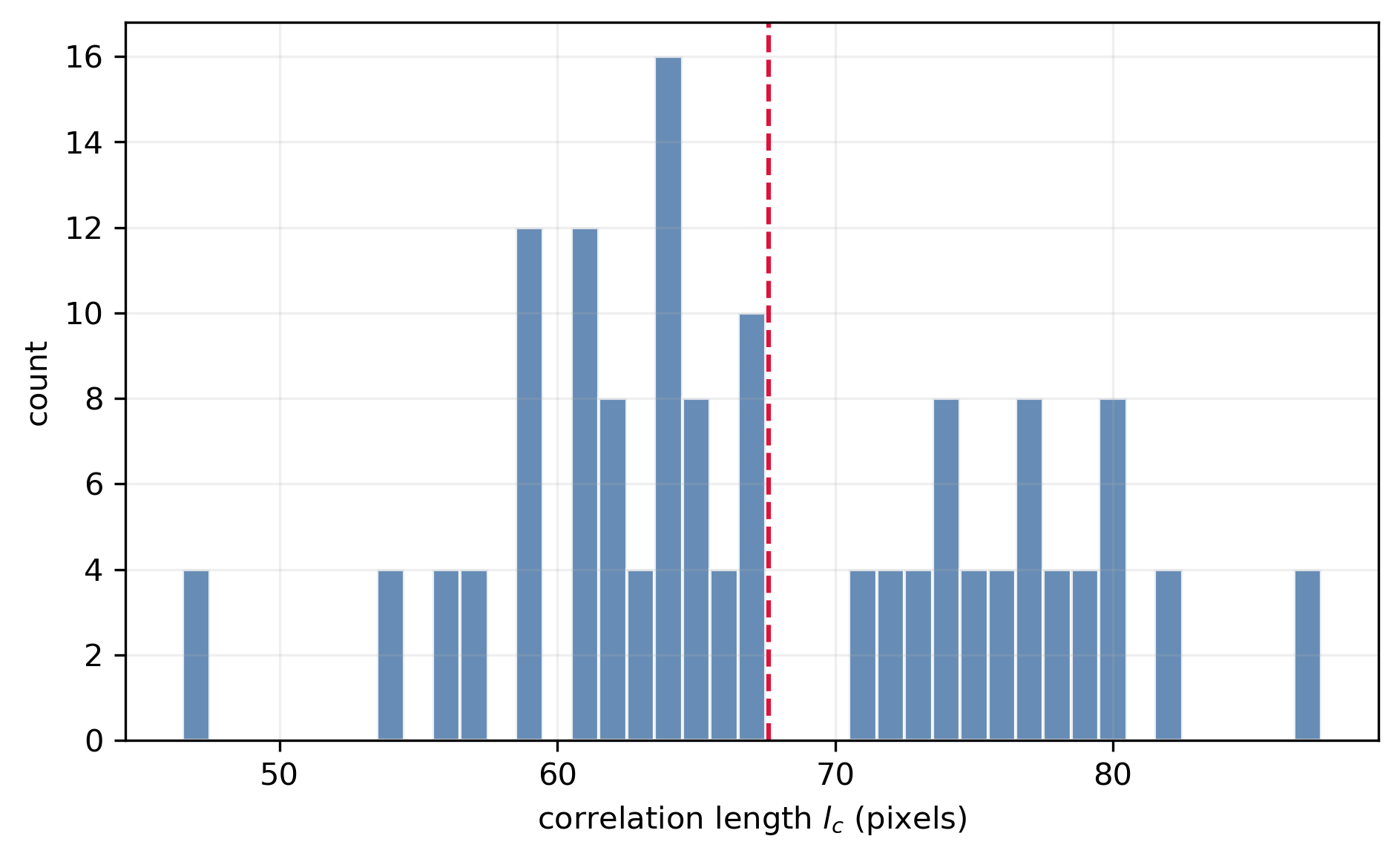}
    \\
    \small\ak{\textbf{(a)} Radial ACF profile}
    &
    \small\ak{\textbf{(b)} Correlation-length histogram}
    \end{tabular}
    \caption{Autocorrelation-based resolution diagnostics for the OAM code $1101$ \ak{under the default setting $(z,\sigma_0,l_0)=(5,1.1,1.5)$}. \ak{Panel (a) shows the empirical mean and standard deviation of the radial ACF over $150$ realizations. Panel (b) shows the histogram of the correlation lengths defined by~\eqref{eq:corrlen_isml}.}}
    \label{fig:acf_validation_isml}
\end{figure}

\begin{table}[t]
\centering
\small
\setlength{\tabcolsep}{4pt}
\caption{Summary of the ACF-based correlation-length estimate for the OAM code $1101$. The correlation length $l_c$ is defined by ~\eqref{eq:corrlen_isml} and reported in pixels.}
\label{tab:acf_validation_isml}
\begin{tabular}{lccccc}
\toprule
\textbf{Code} & \textbf{Label} & \textbf{$n$} & \textbf{mean} & \textbf{std} & \textbf{range} \\
\midrule
1101 & 13 & 150 & \ak{67.61} & \ak{8.85} & \ak{$[47,87]$} \\
\bottomrule
\end{tabular}
\end{table}

\paragraph{Plane-wave scintillation diagnostics.}
To quantify the strength of intensity fluctuations induced by the random medium, we compute the scintillation index function,
\begin{equation}\label{eq:scintillation_isml}
    S(x):=\frac{\mathbb{E}[I^2(x)]-(\mathbb{E}[I(x)])^2}{(\mathbb{E}[I(x)])^2},
\end{equation}
where $\mathbb{E}$ denotes expectation over turbulence realizations and $I(x)=|u(z,x)|^2$ is the received intensity.
In order to interpret ~\eqref{eq:scintillation_isml} away from pixels where $\mathbb{E}[I(x)]$ is negligibly small, we define the region of interest (ROI)
\begin{equation}\label{eq:scintillation_roi_isml}
    \Omega_\tau:=\Big\{x:\mathbb{E}[I(x)]\ge \tau \max_{x'}\mathbb{E}[I(x')]\Big\},\qquad \tau=0.10,
\end{equation}
and report scintillation statistics restricted to $\Omega_\tau$.
For turbulence-strength diagnostics, we use a plane-wave initial condition rather than a structured beam, since in the latter case the scintillation index can be artificially amplified on regions where the mean intensity is very small. \ak{Fig.~\ref{fig:scintillation_codes_isml}} shows a representative propagated plane-wave intensity, the corresponding empirical mean intensity, and the ROI-restricted scintillation map for the current default setting \ak{$(z,\sigma_0,l_0)=(5,1.1,1.5)$}, computed from $20$ realizations. Table~\ref{tab:scintillation_codes_isml} reports the corresponding ROI statistics. In this setting the mean intensity is essentially uniform, so the ROI occupies the full field; the resulting plane-wave scintillation level is therefore a stable indicator of the turbulence strength. Quantitatively, the ROI scintillation is of order one but slightly below the fully developed benchmark, with \ak{mean $0.8378$ and median $0.7770$}.

\begin{figure}[htbp!]
    \centering
    \includegraphics[width=\linewidth]{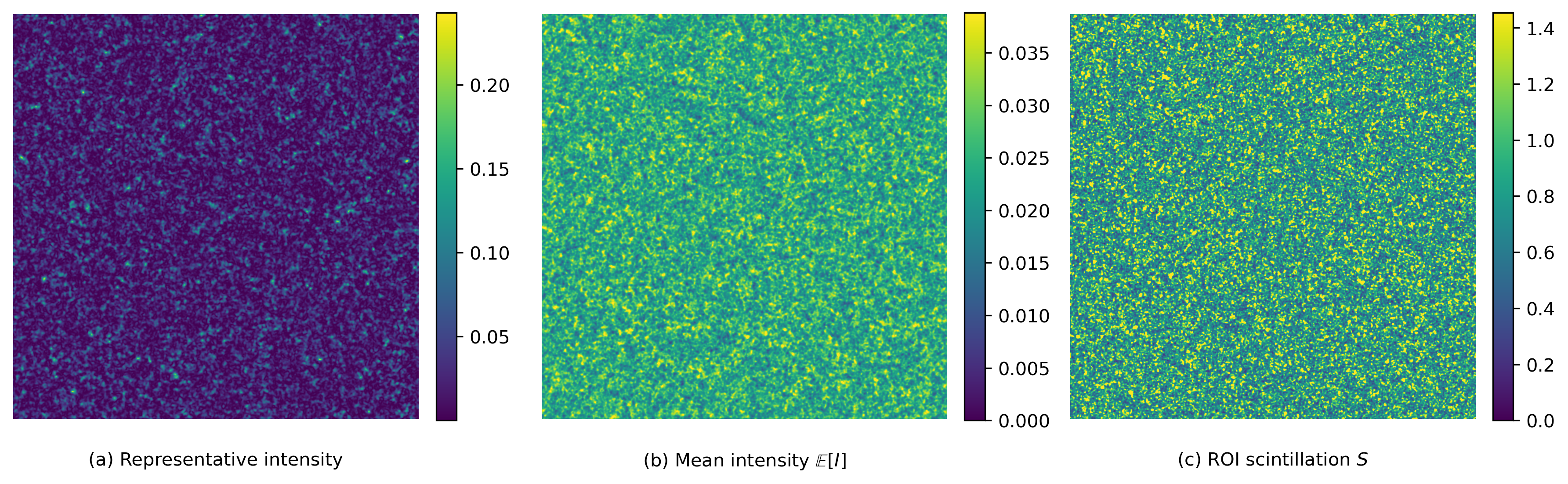}
    \caption{Plane-wave scintillation diagnostics, computed from $20$ realizations. \ak{Panels (a)--(c) show a representative propagated intensity, the empirical mean intensity $\mathbb{E}[I(x)]$, and the scintillation map restricted to the ROI $\Omega_{0.1}$, respectively.}}
    \label{fig:scintillation_codes_isml}
\end{figure}

\begin{table}[t]
\centering
\small
\setlength{\tabcolsep}{4pt}
\caption{Plane-wave scintillation statistics on the ROI $\Omega_\tau$ in ~\eqref{eq:scintillation_roi_isml}. Here $|\Omega_\tau|/|\Omega|$ is the fraction of pixels in the ROI; $\overline{S}_{\Omega_\tau}:=|\Omega_\tau|^{-1}\sum_{x\in\Omega_\tau}S(x)$ denotes the ROI-averaged scintillation index; and the remaining columns report quantiles of $S(x)$ over $x\in\Omega_\tau$.}
\label{tab:scintillation_codes_isml}
\begin{tabular}{lcccccc}
\toprule
\textbf{Input} & \textbf{$n$} & $|\Omega_\tau|/|\Omega|$ & $\overline{S}_{\Omega_\tau}$ & $\mathrm{median}$ & $q_{0.25}$ & $q_{0.75}$ \\
\midrule
plane wave & 20 & $1.0000$ & \ak{$0.8378$} & \ak{$0.7770$} & \ak{$0.6091$} & \ak{$0.9949$} \\
\bottomrule
\end{tabular}
\end{table}

\subsection{Supplementary Classification Comparisons.}\label{app:classification_comparisons}

\paragraph{\ak{Class-wise Errors for the Classifier.}}
\ak{To examine whether the ResNet-18 accuracy is dominated by a small subset of classes, Fig.~\ref{fig:resnet18_intensity_confusion_isml} reports the row-normalized confusion matrix for the recommended intensity-input classifier.
The matrix is computed on the held-out real test set by pooling predictions from three independently trained models under the same protocol as the main accuracy table.
The pooled matrix contains $2025$ test predictions and has overall accuracy $98.32\%$.}

\begin{figure}[htbp!]
    \centering
    \includegraphics[width=0.82\linewidth]{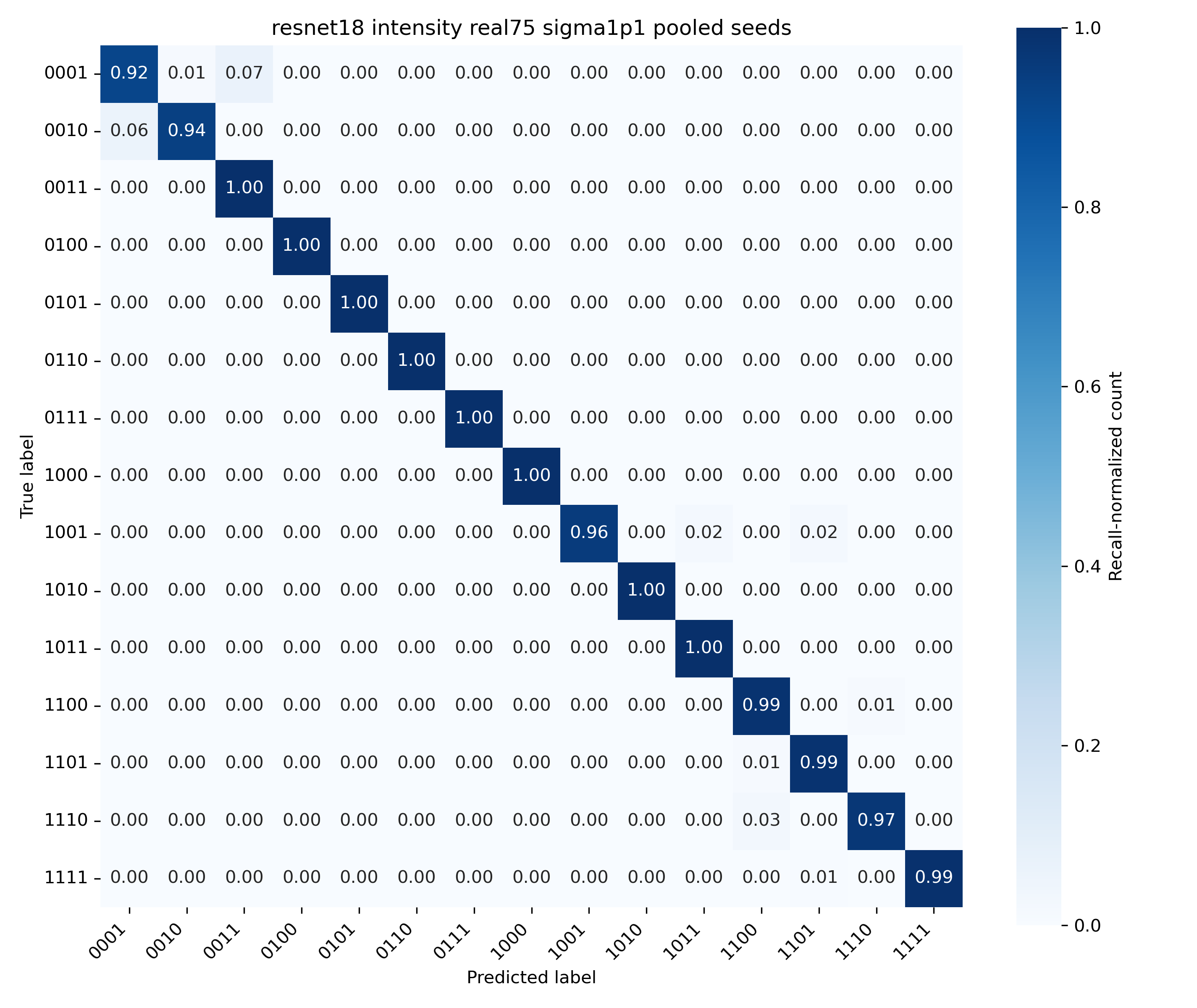}
    \caption{\ak{Row-normalized confusion matrix for the recommended ResNet-18 classifier with cropped intensity input under the default setting $(z,\sigma_0,l_0)=(5,1.1,1.5)$ and $N_{\mathrm{train}}=75$ real training samples per class. The matrix is computed on the held-out real test split by pooling predictions over the three independently trained models with seeds $\{42,100,2023\}$; the corresponding pooled overall accuracy is $98.32\%$.}}
    \label{fig:resnet18_intensity_confusion_isml}
\end{figure}

\paragraph{\ak{Alphabet-size Scaling.}}
\ak{We also tested how the classifier output layer scales when the discrete source alphabet is changed.
For a $C$-class alphabet, the convolutional feature extractor is unchanged; only the final linear layer is replaced by a layer with $C$ logits, and the one-hot label vector is changed from an element of $\{0,1\}^{15}$ to an element of $\{0,1\}^{C}$.
To verify this experimentally, we trained SimpleCNN and ResNet-18 on the first $C$ classes of the source alphabet for $C\in\{5,10,15\}$.
All runs used the same propagated dataset setting: $50$ real training samples per class, the centered $64\times64$ intensity crop, zero padding, and the random seeds $\{42,100,2023\}$.
The model was trained from scratch for the fixed $500$-epoch budget for every architecture, alphabet size, and seed.
No weights were reused from the $15$-class classifier.
The held-out test accuracies in Table~\ref{tab:alphabet_scaling_appendix} show that changing the alphabet size only requires changing the output layer and label dimension, although the classification accuracy may vary with the number and separability of the classes.}

\begin{table}[htbp!]
\centering
\small
\caption{\ak{Alphabet-size scaling study for the classifier architectures. Entries are top-1 test accuracies (\%) reported as mean $\pm$ standard deviation over the three independently trained seeds $\{42,100,2023\}$.}}
\label{tab:alphabet_scaling_appendix}
\begin{tabular}{lccc}
\toprule
\ak{\textbf{Model}} & \ak{\textbf{$C=5$}} & \ak{\textbf{$C=10$}} & \ak{\textbf{$C=15$}} \\
\midrule
\ak{SimpleCNN} & \ak{$97.04 \pm 0.93$} & \ak{$91.41 \pm 2.69$} & \ak{$91.60 \pm 0.31$} \\
\ak{ResNet-18} & \ak{$95.56 \pm 1.60$} & \ak{$93.70 \pm 0.46$} & \ak{$93.38 \pm 1.85$} \\
\bottomrule
\end{tabular}
\end{table}

\ak{The parameter-count change is confined to the final classifier layer.
In this experiment, SimpleCNN changes from $93{,}765$ trainable parameters for $C=5$ to $95{,}055$ for $C=15$, while ResNet-18 changes from $11{,}172{,}805$ to $11{,}177{,}935$ trainable parameters.
Thus almost all trainable parameters are shared across alphabet sizes.}

\paragraph{Random Shifts.}
This paragraph supplements Section 3.E in \ak{the} main text by comparing deterministic fixed-shift training with random-shift training at the same shift magnitude $S$, always evaluated under the central-window protocol.
Under the present protocol, random-shift training is markedly more robust than deterministic fixed-shift training once $S>0$.
Table~\ref{tab:shift_protocol_appendix} reports the corresponding numerical results.
\begin{table}[htbp!]
\centering
\small
\setlength{\tabcolsep}{4pt}
\caption{Fixed-shift versus random-shift training under the default setting, with centered testing in all cases.}
\label{tab:shift_protocol_appendix}
\begin{tabular}{lccc}
\toprule
\textbf{Protocol} & \textbf{$S$} & \textbf{SimpleCNN} & \ak{\textbf{ResNet-18}} \\
\midrule
fixed  & 16 & \ak{$51.65\pm1.33$} & \ak{$8.84\pm4.05$} \\
random & 16 & \ak{$79.51\pm3.77$} & \ak{$93.73\pm1.13$} \\
fixed  & 32 & \ak{$25.23\pm3.41$} & \ak{$15.51\pm7.89$} \\
random & 32 & \ak{$64.05\pm5.69$} & \ak{$91.41\pm2.57$} \\
fixed  & 48 & \ak{$9.14\pm1.19$} & \ak{$0.44\pm0.15$} \\
random & 48 & \ak{$53.98\pm4.46$} & \ak{$88.15\pm3.60$} \\
\bottomrule
\end{tabular}
\end{table}

\begin{figure*}[t]
    \centering
    \includegraphics[width=\textwidth]{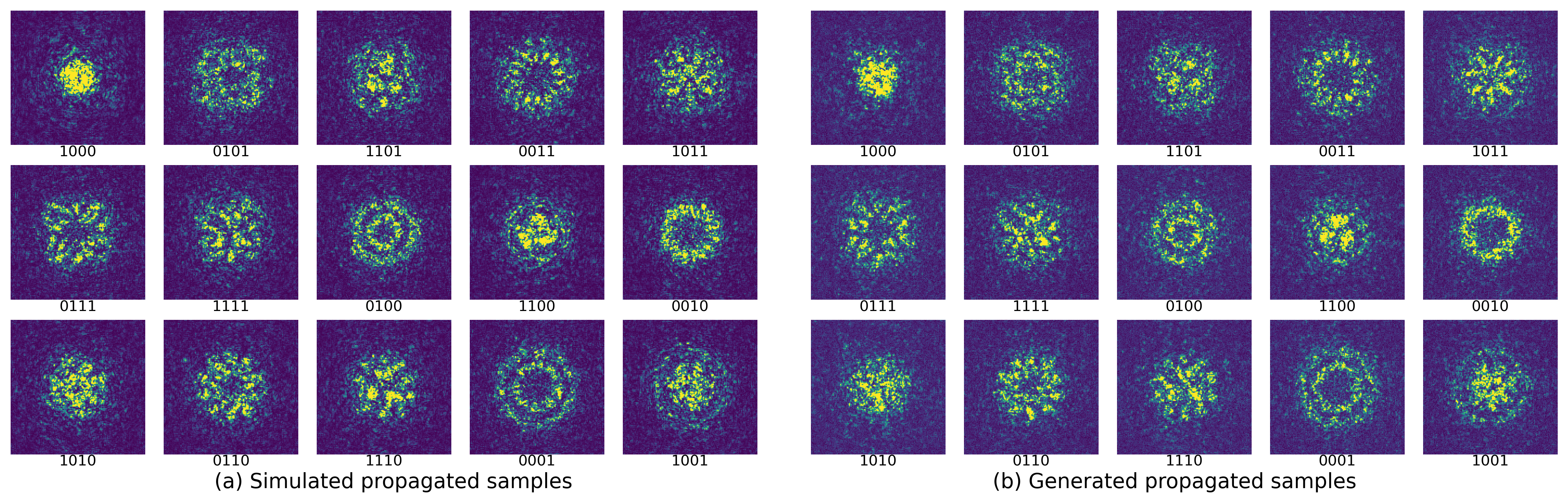}
    \caption{\ak{Dataset-overview-style comparison} of simulated and generated propagated patterns \ak{under the default setting $(z,\sigma_0,l_0)=(5,1.1,1.5)$ and the $\mathbf{v}/\mathbf{v}$ configuration with $\lambda=1$}. \ak{Panel (a) shows one representative simulated propagated sample for each of the $15$ classes, after the same $2048\to256$ average-pooling step used in the learning pipeline. Panel (b) shows one representative generated sample for each corresponding class.} \ak{Both panels} are displayed with the common intensity range $[0,0.15]$. In particular, the code $1101$ appears in the third position of the first row in \ak{both panels}.}
    \label{fig:gen_vs_real_appendix}
\end{figure*}

\paragraph{Convolutional Padding Strategies.} 
Although the SSFM simulation uses a periodic Fourier representation at the full-field level, the classifier operates on localized crops.
Consequently, circular padding at crop boundaries may create artificial wrap-around correlations by identifying opposite crop edges.
Unless explicitly stated otherwise, all classification experiments in this paper use zero padding.
We compare this default choice with circular padding in a separate controlled comparison.
Here ``padding'' refers to the boundary condition used by each convolutional layer on the \emph{cropped} input (and intermediate feature maps), rather than any padding applied to the physical intensity measurements. Table~\ref{tab:padding_comparison_isml} shows only a minor dependence on the convolutional boundary condition at the crop scale.

\begin{table}[htpb!]
\centering
\small
\setlength{\tabcolsep}{4pt}
\caption{Controlled comparison of padding strategies for the dataset generated with \ak{$(z,\sigma_0,l_0)=(5,1.1,1.5)$}, in the no-shift setting $S=0$, where training, validation, and testing all use the centered crop $\delta=(0,0)$ on a $64\times64$ crop ($N_{\mathrm{train}}=50$).}
\label{tab:padding_comparison_isml}
\begin{tabular}{lcc}
\toprule
\textbf{Model} & \textbf{Padding $p$} & \textbf{Accuracy (\%)} \\
\midrule
\multirow{2}{*}{SimpleCNN}
 & zeros    & \ak{$92.59\pm1.27$} \\
 & circular & \ak{$90.12\pm0.99$} \\
\midrule
\multirow{2}{*}{\ak{ResNet-18}}
 & zeros    & \ak{$93.48\pm1.18$} \\
 & circular & \ak{$95.11\pm0.30$} \\
\bottomrule
\end{tabular}
\end{table}

\subsection{Supplementary Generative-model Results.}\label{supp:gen_aug}

\paragraph{Visual Comparisons.}
\ak{Fig.~\ref{fig:ablation_grid}} displays representative generated outputs for the diffusion-model configurations considered in the main text, under a common display range, so that the qualitative effect of changing the prediction/loss parameterization can be compared directly.
\ak{Fig.~\ref{fig:gen_vs_real_appendix} then provides a dataset-overview-style comparison} of the learned codebook: \ak{the simulated propagated samples are shown alongside the corresponding generated outputs under the $\mathbf{v}/\mathbf{v}$ configuration with $\lambda=1$ for the updated default setting $(z,\sigma_0,l_0)=(5,1.1,1.5)$}.

\begin{figure}[t!]
    \centering
    \includegraphics[width=0.99\linewidth]{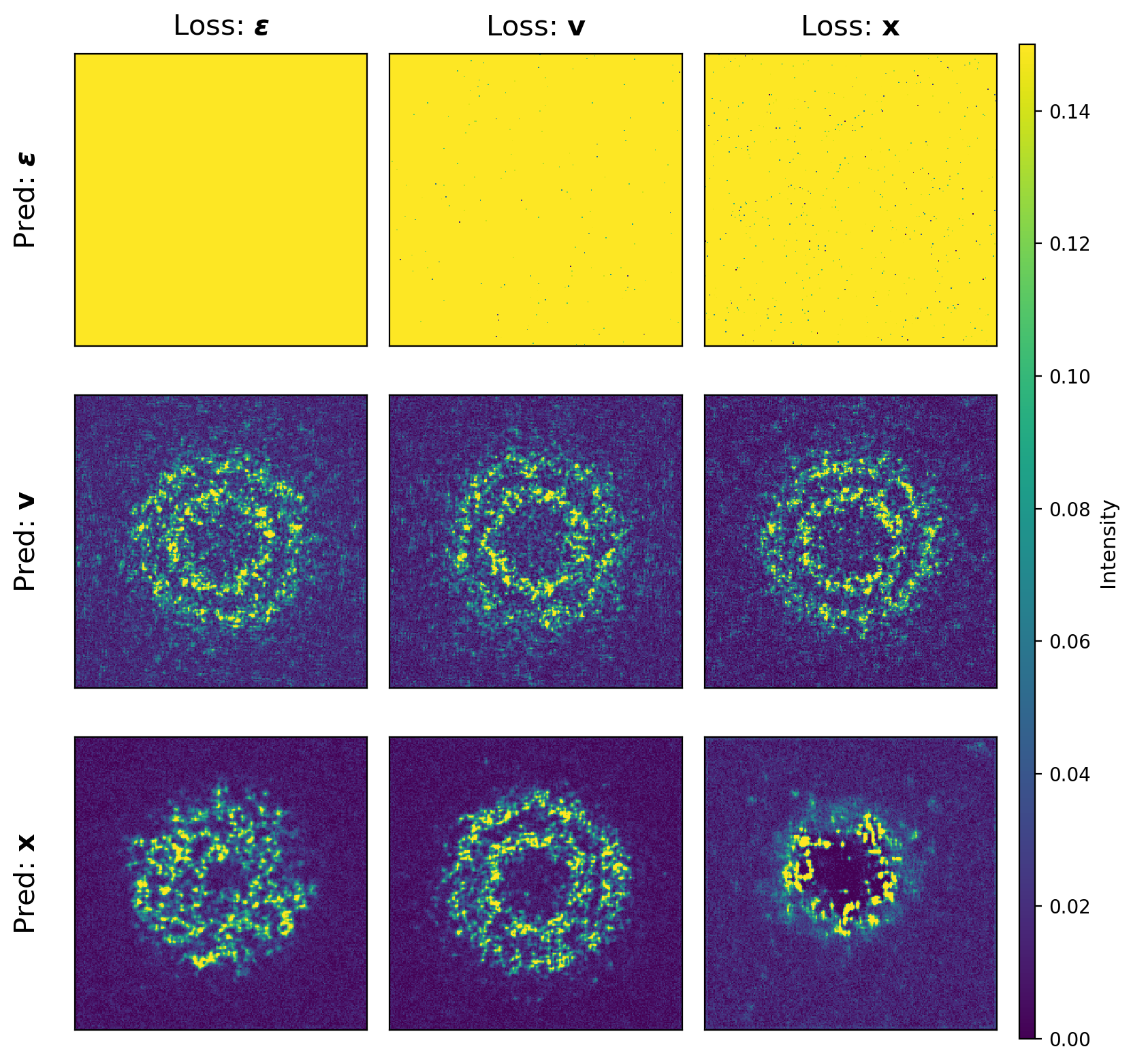}
    \caption{Representative generated samples across the controlled diffusion-model configurations from Section 4.E in \ak{the} main text. \ak{The $3\times3$ grid compares diffusion prediction targets by row and loss targets by column.} All \ak{images} are displayed with the common intensity range $[0,0.15]$.}
    \label{fig:ablation_grid}
\end{figure}

\end{document}